\documentclass[twoside,11pt]{article}

\usepackage{jmlr2e}

\usepackage{subfigure}
\usepackage{verbatim}

\usepackage{amsmath,amsfonts}
\usepackage{url}
\newtheorem{assumption}{Assumption}

\newcommand{\bitm}{\begin{itemize}}
\newcommand{\eitm}{\end{itemize}}
\newcommand{\beqa}{\begin{eqnarray}}
\newcommand{\eeqa}{\end{eqnarray}}
\newcommand{\beqas}{\begin{eqnarray*}}
\newcommand{\eeqas}{\end{eqnarray*}}
\newcommand{\baln}{\begin{align}}
\newcommand{\ealn}{\end{align}}
\newcommand{\balns}{\begin{align*}}
\newcommand{\ealns}{\end{align*}}

\newcommand{\probSimplex}[1] {\mathcal{P}\parenth{#1}}

\newcommand{\kldist}[2] {D\! \parenth{#1\|#2}}

\newcommand{\alphabet}[1] { {\mathsf #1}}

\newcommand{\parenth}[1] {\left(#1\right)}

\newcommand{\var}{\mathrm{var}}
\newcommand{\I}{\mathrm{I}}

\renewcommand*{\H}{\text{H}}

\newcommand{\D}{\mathrm{D}}

\newcommand{\X}{\bX}
\newcommand{\x}{\mathbf{x}}
\newcommand{\Y}{\mathbf{Y}}
\newcommand{\Z}{\mathbf{Z}}

\newcommand{\N}{\mathbf{N}}

\newcommand{\uX}{\underline{X}}
\newcommand{\ux}{\underline{x}}

\def\argmin{\mathop{\arg\,\!\min}\limits}%
\def\argmax{\mathop{\arg\,\!\max}\limits}%

\newcommand{\calX}{\alphabet{X}}

\newcommand{\calO}{\mathcal{O}}
\newcommand{\calS}{\mathcal{S}}
\newcommand{\bX}{\mathbf{X}}

\newcommand{\bx}{\mathbf{x}}

\newcommand{\allX}{\underline{\bX}}

\newcommand{\allx}{\underline{\bx}}

\newcommand{\uA}{\underline{\mathrm{A}}}

\newcommand{\uW}{\underline{\mathrm{W}}}

\newcommand{\PT}{\widehat{P}_{\allX}}

\newcommand{\calG}{\mathcal{G}}
\newcommand{\setmi}[2]{  [#1]\backslash\{#2\}}
\newcommand{\calB}{\mathcal{B}}
\newcommand{\calT}{\mathcal{T}}

\newcommand{\Phat}{\widehat{P}}

\newcommand{\tildeA}{\widetilde{A}}
\newcommand{\wtA}{\widetilde{A}}
\newcommand{\wtB}{\widetilde{B}}

\newcommand{\std}{\mathrm{std}}

\newcommand{\denote}{:=}

\newcommand{\titlename}{Bounded Degree Approximations of Stochastic Networks}

\ShortHeadings{\titlename}{Quinn, Pinar, and Kiyavash}

\firstpageno{1}

\begin{document}

\title{\titlename}

\author{\name Christopher J.\ Quinn \email cjquinn@purdue.edu \\
       \addr School of Industrial Engineering\\
       Purdue University\\
       West Lafayette, Indiana 47907, USA
       \AND
       \name Ali Pinar \email apinar@sandia.gov \\
       \addr Data Science \& Cyber Analytics Department\\
       Sandia National Laboratories\\
       Livermore, CA 94551, USA
			 \AND
       \name Negar Kiyavash \email kiyavash@illinois.edu \\
       \addr Department of Industrial and Enterprise Systems Engineering\\
       University of Illinois\\
      Urbana, Illinois 61801, USA}

\editor{}

\maketitle

\begin{abstract}
We propose algorithms to approximate directed information graphs.  Directed information graphs are probabilistic graphical models that depict causal dependencies between stochastic processes in a network.  The proposed algorithms  identify optimal and near-optimal approximations in terms of Kullback-Leibler divergence.  The user-chosen sparsity trades off the quality of the approximation against visual conciseness and computational tractability.  One class of approximations contains graphs with specified in-degrees.  Another class additionally requires that the graph is connected.  For both classes, we propose algorithms to identify the optimal approximations and also near-optimal approximations, using a novel relaxation of submodularity.  We also propose algorithms to identify the $r$-best approximations among these classes, enabling robust decision making. 
\end{abstract}


\section{Introduction}

Many fields of the sciences and engineering require analysis, modeling, and decision making using networks, typically represented by graphs.  Social networks, financial networks, and biological networks are a few categories that are relevant not only academically but also in daily life.  A major challenge for studying networks is identifying a concise topology, such as who strongly influences whom in a social network.  Real world networks are often large---there are tens of thousands of human genes and trillions of connections in the human brain.  Such scales make human visual processing of and decision making with the whole network prohibitive.  This paper investigates algorithms to identify provably good approximations of the network topology, which capture important system dynamics while significantly reducing the number of edges  to enable tractable analysis. 
  
Across different domains, edges  are used to model various kinds of ties such as physical connections or dynamic relationships. For instance, in depicting a computer network, an edge might correspond to a physical wire or a packet exchange between a sender and a receiver. For the human brain, edges could correspond to information flow between different cells or brain regions \citep{takahashi2015large, kim2014dynamic}.  For online social networks, edges might represent user-defined relationships or pairs of users that frequently message each other \citep{VerSteeg2012information, ver2013information}.  Edges that represent dynamics often must be inferred from activity, in some cases statistically.  

There is a large literature on graphs whose edges depict statistical relationships.  Markov and Bayesian networks are well-known probabilistic graphical models whose edges represent correlation.  Many methods have been proposed to infer and approximate the networks from i.i.d.\ data, often relying on heuristics.  For an overview, see Chapters~18~and~20 in \cite{koller2009probabilistic}.  For applications involving agents interacting with each other over time, whether in finance, biology, social networks, or other domains, there is interest in identifying and representing causal influences between the agents, not just correlation.

One approach uses known families of models, such as with structural equation modeling, to distinguish cause and effect \citep{ZhaZhaSch15, ZhangWZS2014, ChenZC2012}.  A recent work uses belief propagation to infer directionality \citep{chang2014causal}.  Alternatively, under appropriate conditions such as with expert labeling and no feedback, Bayesian networks can depict causal relationships using Pearl's interventional calculus \citep{pearl2009causality}.   We consider the general setting when such conditions or modeling assumptions might not hold.

Recently, directed information graphs were introduced to address this issue \citep{quinn2011estimating, amblard2011directed}.  Edges in directed information graphs depict statistical causation between non-i.i.d.\ time-series. In this work, ``statistical causation'' is in the sense of Granger causality \citep{granger1969investigating}, where a process $\X$ statistically causes $\Y$ if in sequentially predicting $Y_t$, knowledge of the past $X^{t-1}$ helps in prediction even when $Y^{t-1}$ and the past of all the other processes are already known.  These graphs use directed information, an information theoretic quantity, which is well-defined for any class of stochastic processes.  Directed information has been applied to a range of settings, such as neuroscience \citep{quinn2011estimating, kim2011Granger,so2012assessing, kim2014dynamic}, gene regulatory networks \citep{rao2007motif,rao2008using}, and online social networks \citep{VerSteeg2012information,ver2013information, quinn2012directed}.

For networks with thousands or millions of edges, directed information graphs become too  complicated for direct humans analysis.  A major approach to simplifying the graphs is to only keep a few edges which together best approximate the dynamics of the system.  For example, a directed tree is among the simplest graphs.  See Figure~\ref{fig:diag:tree_in2}.  Each node has only  one parent.  Trees have the fewest number of edges possible while being connected.  There is a root node and a path from the root to every other node.  The graph is concise, facilitating human analysis and decision making.  A recent work proposed an efficient algorithm to identify the best directed tree approximation, where   goodness of approximation is measured by Kullback-Leibler (KL) divergence from the full joint distribution to the distribution induced by the directed tree \citep{quinn2013efficient}.  In addition to being computationally efficient, the algorithm in \cite{quinn2013efficient} only uses joint statistics for pairs of processes and  does not require the full joint distribution to find the best approximation.  

\begin{figure}[t]
\centering
  \subfigure[A directed tree.]{\label{fig:diag:tree} \includegraphics[width=.30\columnwidth]{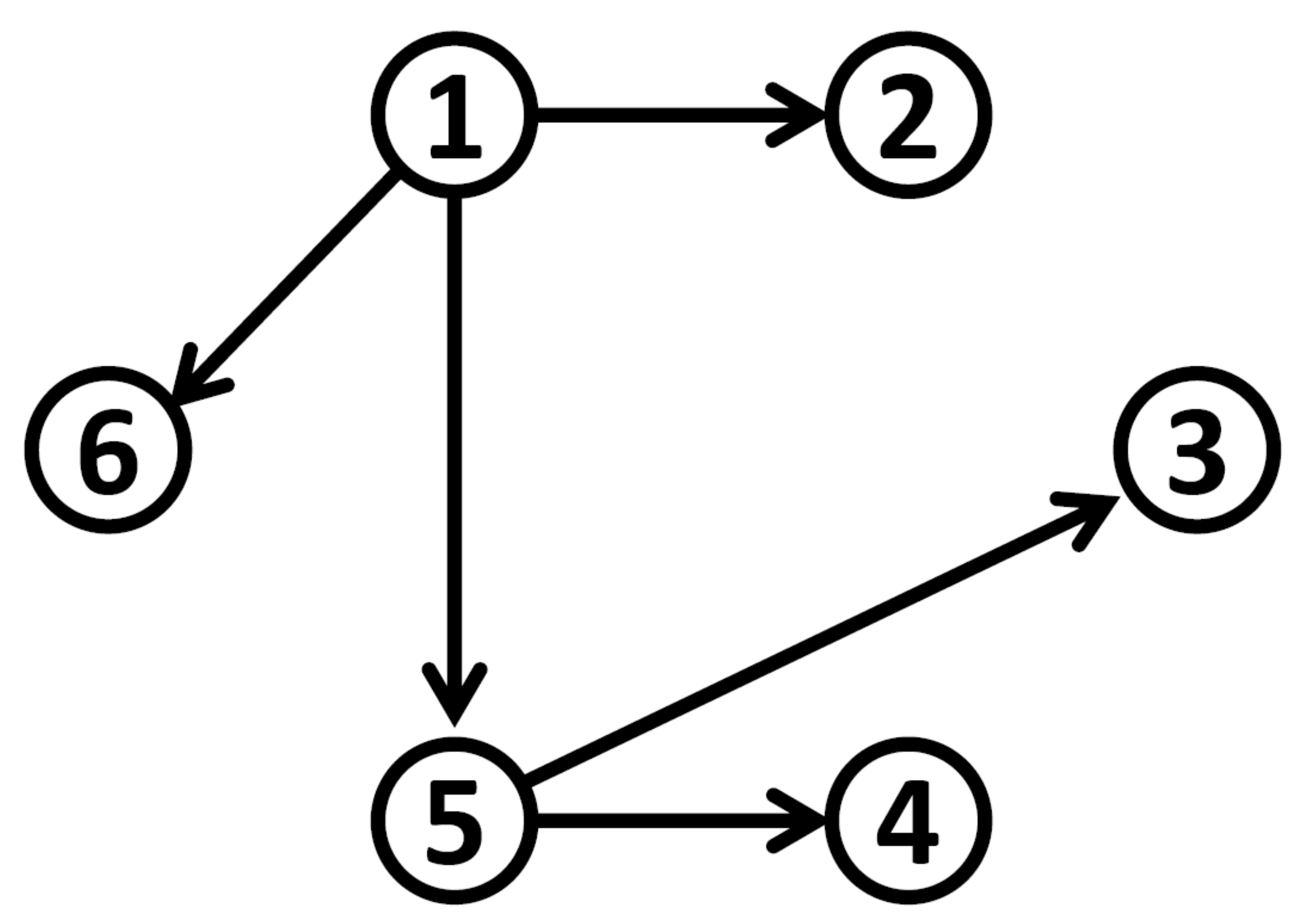}} 
  \hspace{0.1cm}
  \subfigure[A graph with in-degree two containing a directed tree.]{\label{fig:diag:tree2in} \includegraphics[width=.30\columnwidth]{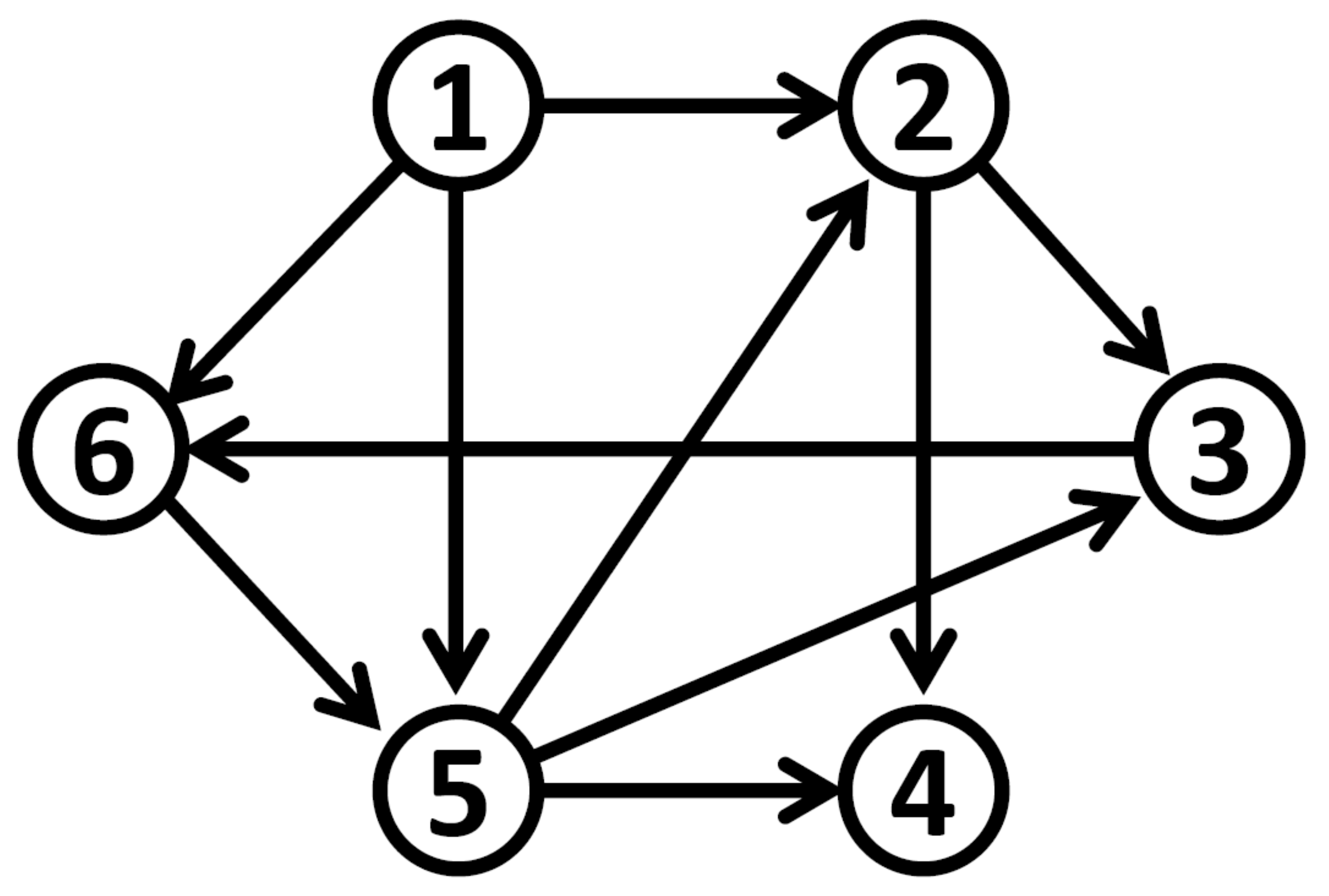}}  
  \caption{\small Diagrams for two network approximations. 
	}
  \label{fig:diag:tree_in2}
\end{figure}

Though directed tree approximations are easy to comprehend and efficient to construct, they cannot depict feedback. Feedback is essential in many networks, such as in the brain and gene regulatory networks.  Thus, for some applications, it is necessary to consider higher order approximations.  For instance, a graph with in-degrees two and three and containing a directed spanning tree as a subgraph would trade-off some simplicity and computational efficiency in order to capture more complex relationships in the network.  

\subsection{Our Contributions}

We propose an algorithm to identify the optimal connected bounded in-degree approximations.  The algorithm requires only low-dimensional statistics, similar to the algorithm for directed tree approximations.  The user decides how  complex to make the approximations, changing the in-degrees to trade off visual and computational simplicity against the accuracy of the approximation.  

Identifying optimal approximations becomes prohibitive for large in-degrees.  For situations where a near-optimal approximation would suffice, we propose algorithms using a greedy search.  We identify sufficient conditions, namely a relaxed form of submodularity, that ensure near-optimality.  

Additionally, having multiple, good approximations can aid in understanding network dynamics.  Instead of just having the best approximation, having the five or ten best approximations in order can yield insight into which edges are most important---those that persist in the top approximations---and those that are less significant.  Being able to identify the top-$r$ approximations also enables the user to identify the best approximation of more restricted classes of topologies.  For example, suppose that the best directed tree approximation for a network had a height of six.  If the user desires the best directed tree approximation with height less than four, he/she can look among the top-$r$ approximations until  finding a tree with height less than four and it would necessarily be the best such approximation.  We develop algorithms to identify the top-$r$ approximations with similar complexity as finding the optimal approximation.

Lastly, we use simulations to validate the quality of the approximations found.

\subsection{Related Work}


There is a large body of work on approximating Bayesian and Markov networks.  One well known result is an algorithm to identify optimal tree approximations \citep{chow1968approximating}.  The algorithm finds a maximum weight spanning tree using mutual information for weights and only requires distributions of pairs of variables.   

In general, identifying more complex approximations cannot be done in a computationally efficient manner.   Bayesian networks are NP-hard to approximate for topologies with specified in-degree larger than one \citep{chickering1996learning} and even polytrees with in-degree two \citep{dasgupta1999learning}.  Some works have focused on identifying optimal approximations of subclasses of polytrees.  One work finds the best bounded in-degree approximation that preserves the statistical dependencies in the best tree approximation \citep{carvalho2007learning}.  Another work finds an optimal polytree that can be converted to a tree with a bounded number of edge or node deletions  \citep{gaspers2012on}.  

Other approaches to approximating graphical models include using $l_1$-regularized regression to identify sparse Ising models for Markov networks with binary variables \citep{ravikumar2010high}.  Another approach proposes a linear programming relaxation coupled with branch and bound to find an optimal approximation \citep{jaakkola2010learning}.  Annealed importance sampling is used in \cite{niinimaki2013annealed}; see references therein for Markov chain Monte Carlo based techniques.  The performance of a forward-backward greedy search for Markov networks in a high-dimensional setting is studied in \cite{jalali2011}.  In \cite{pernkopf2010efficient}, an algorithm is proposed to first identify an variable ordering and then greedily select parents.

There has been much less work developing approximations for directed information graphs.  In \cite{quinn2013efficient}, an algorithm is proposed to identify the best directed spanning tree approximation for directed information graphs.  In \cite{quinn2012directed}, several algorithms are introduced for inferring the exact topology.  One of the algorithms can be also used to compute the best approximation where the only topological constraints are user-specified in-degrees.  That is discussed here as Algorithm~1 in Section~\ref{sec:optbnd}.  Several works investigated sparse approximations using lasso and related penalties when processes are jointly autoregressive with Gaussian noise \citep{charbonnier2010weighted, haufe2010sparse, bolstad2011causal,jung2013compressive, basu2012network}.

In our preliminary work \cite{quinn2013optimal}, we developed an algorithm to identify the optimal bounded in-degree approximation containing a directed spanning tree subgraph.  This appears here as Algorithm 2.   Also, a sufficient condition for a greedy search to return near-optimal approximations was identified in \cite{quinn2013optimal}, presented here as Definition~\ref{def:grdsub}.

There has been research in the graphical models literature for finding the top-$r$ solutions for problems such as the MAP realizations for Bayesian or Markov networks \citep{nilsson1998efficient, yanover2003finding, fromer2009lp, flerova2012bucket, batra2012diverse}.  The present work focuses on finding the top-$r$ solutions for structure learning.

\subsection{Paper Organization}
The paper is organized as follows.  Definitions and notations are introduced in Section~\ref{sec:defns}.  Section~\ref{sec:DIrev} reviews directed information graphs. Section~\ref{sec:optbnd} presents algorithms to identify the optimal bounded in-degree approximations.  Section~\ref{sec:eff_near_opt_con} identifies a sufficient condition for the greedy search to construct near optimal approximations.  Section~\ref{sec:best_r} describes an algorithm to find the top-$r$ approximations. Algorithmic complexity is discussed in Section~\ref{sec:complexity}.  The algorithms are empirically evaluated in Section~\ref{sec:sims}.    Section~\ref{sec:concl} concludes the paper.  Proofs are in the appendix.

\section{Notation and Information-Theoretic Definitions}\label{sec:defns}
We now define notation.  We use ``\denote'' for denoting.
\begin{itemize} \setlength\itemsep{0.5em}
\item For a sequence $a_1,a_2,\ldots$, denote $(a_i,\ldots,a_j)$ as $a_i^j$ and $a^k := a_1^k$. Let $[m] := \{1,\ldots,m\}$ and the power set $2^{[m]}$ on $[m]$ to be the set of all subsets of $[m]$.

\item We consider $m$ finite-alphabet, discrete-time random processes over a horizon $n$.  Let $\calX$ denote the alphabet and $\probSimplex{\calX}$ the space of probability measures on $\calX$. Denote the $i$th random variable at time $t$ by $X_{i,t}$, 
the $i$th random process as $\bX_i = (X_{i,1},\ldots,X_{i,n})^{\top}$, the whole collection of all $m$ random processes as $\allX= (\bX_1,\ldots,\bX_m)^{\top}$, and a subset of $K$ processes indexed by $A \subseteq [m]$ as $\allX_{A}= (\bX_{A(1)},\ldots,\bX_{A(K)})^{\top}$.

\begin{remark} We consider the finite-alphabet setting to simplify the presentation.  The results extend to more general cases.
\end{remark}

\item Conditional and {\it causally conditioned} distributions \citep{kramer1998directed} of $\X_i$ given $\X_j$ are %
\beqa %
P_{\X_i|\X_j}(\x_i|\x_j)  &:=& \prod_{t=1}^n P_{X_{i,t}|X_{i}^{t-1},X_j^n}(x_{i,t}|x_{i}^{t-1},x_j^n)  \label{eq:def:cond_distr}\\
P_{\X_i\|\X_j}(\x_i\|\x_j) &:=& \prod_{t=1}^n P_{X_{i,t}|X_{i}^{t-1},X_j^{t-1}}(x_{i,t}|x_{i}^{t-1},x_j^{t-1}). \label{eq:def:causal_cond}
\eeqa
Note the similarity between \eqref{eq:def:cond_distr} and \eqref{eq:def:causal_cond}, though in \eqref{eq:def:causal_cond} the present and future, $x^n_{j,t}$, is not conditioned on.  In \cite{kramer1998directed}, the present $x_{j,t}$ was conditioned on in \eqref{eq:def:causal_cond}.  The reason we remove it will be made clear in Remark~\ref{rmrk:cond_fut}.

\item  Consider the set of processes $\allX_A$ for some $A \subseteq \setmi{m}{i}$.  Next consider two sets of causally conditioned distributions
$\{P_{\X_i\| \allX_A=\allx_A} \in \probSimplex{\calX}: \allx_A \in \calX^{|A|n}\}$ and 
$\{Q_{\X_i\| \allX_A=\allx_A} \in \probSimplex{\calX}: \allx_A \in \calX^{|A|n}\}$ along with
a marginal distribution $P_{\allX_A} \in \probSimplex{\calX^{|A|n}}$.  Then the conditional Kullback-Leibler (KL) divergence between causally conditioned distributions is given by
\begin{align}  
& \hspace{-0.5cm}\kldist{P_{\X_i\| \allX_A}}{Q_{\X_i\| \allX_A} | P_{\allX_A}} \nonumber \\
& \hspace{2.5cm}:= \sum_{t=1}^n \sum_{\ux_A^{t-1} 
} \!\kldist{P_{X_{i,t}| \uX_A^{t-1}=\ux_A^{t-1}}}{Q_{X_{i,t}| \uX_A^{t-1}=\ux_A^{t-1}}}  P_{\uX_A^{t-1}}(\ux_A^{t-1}). \label{eq:def:condKL}
\end{align}

\vfill

\item
Let $i,j \in [m]$ and $A \subseteq \setmi{m}{i,j}$.  The mutual information,  {\it directed information} \citep{marko1973bidirectional}, and causally conditioned directed information \citep{kramer1998directed} are 
\begin{align}
\I(\X_j;\X_i) :=& \ \kldist{P_{\X_i,\X_j}}{P_{\X_i} P_{\X_j}} = \kldist{P_{\X_i|\X_j}}{P_{\X_i} | P_{\X_j}} \label{eqn:defn:MutualInformation} \\
=& \sum_{t=1}^n \I(X_j^n; X_{i,t} | X_i^{t-1}) \nonumber \\
\I(\X_j \to \X_i) :=& \  \kldist{P_{\X_i\|\X_j}}{P_{\X_i} | P_{\X_j}}\nonumber 
\\
=& \sum_{t=1}^n \I(X_j^{t-1}; X_{i,t} | X_i^{t-1}) \nonumber\\
\I(\X_j \to \X_i \| \allX_A) :=&  \  \kldist{P_{\X_i\| \allX_{A \cup \{j\}}}}{P_{\X_i\|\allX_A} | P_{\allX_{A \cup \{j\}} }} \label{eqn:defn:ccDirectedInformation} 
\\
=& \sum_{t=1}^n \I(X_j^{t-1}; X_{i,t} | X_i^{t-1},\uX_A^{t-1}). \nonumber
\end{align}
While mutual information quantifies statistical correlation (in the colloquial sense of statistical interdependence), directed information quantifies statistical \emph{causation} in the sense of Granger causality \citep{quinn2012directed, amblard2012relation}.  Note that
$\I(\X_j;\X_i)=\I(\X_i;\X_j)$, but $\I(\X_j \to \X_i) \neq \I(\X_i \to \X_j)$ in general.

\begin{remark} \label{rmrk:cond_fut}
In \eqref{eq:def:causal_cond} and \eqref{eqn:defn:ccDirectedInformation}, there is no conditioning on the present $X_{j,t}$.  This follows Marko's definition \citep{marko1973bidirectional} and is consistent with Granger causality \citep{granger1969investigating}.  \cite{massey1990causality} and \cite{kramer1998directed} later included conditioning on $X_{j,t}$ for the specific setting of communication channels. 
\end{remark}

\end{itemize}

\section{Directed Information Graphs} \label{sec:DIrev}

In this section, we briefly review directed information graphs \citep{quinn2011estimating, amblard2011directed}.  
\begin{definition} A \emph{directed information graph} is a probabilistic graphical model where each node represents a process $\X_i$ and an edge $\X_j \to \X_i$ is drawn if \[ \I(\X_j \to \X_i \| \allX_{[m]\backslash\{i,j\}}) >0.\]
\end{definition}
It follows immediately that directed information graphs are unique for a given distribution $P_{\allX}$.  Under certain conditions, the directed information graph corresponds to a particular factorization of the joint distribution.  By the chain rule, the joint distribution $P_{\allX}$  factorizes over time as $P_{\allX}(\allx)  =  \prod_{t=1}^n P_{ \allX_{t}   |  \allX^{t-1}  } (  \allx_{t} |  \allx^{t-1}  ). $ If given the full past $\allX^{t-1}$, the processes $\{\X_1, \dots, \X_m\}$ at time $t$ are mutually independent, $P_{\allX}$ can be further factorized as
\begin{eqnarray}
P_{\allX}(\allx) \!\!\!&=&\!\!\!  \prod_{t=1}^n  \prod_{i=1}^m P_{X_{i,t}| \allX^{t-1}} (x_{i,t}| \allx^{t-1} ),  \label{eq:MGM:2}
\end{eqnarray} and $P_{\allX}$ is said to be \emph{strictly causal}.  Equation~\ref{eq:MGM:2} can be written using causal conditioning notation \eqref{eq:def:causal_cond} as
$ 
P_{\allX}(\allx) =  \prod_{i=1}^m  P_{\bX_i \parallel \allX_{\setmi{m}{i}} } (\bx_i \parallel \allx_{\setmi{m}{i}}). 
$ A distribution $P_{\allX}$ is said to be \emph{positive} if $P_{\allX}(\allx) > 0$ for all $\allx \in \calX^{mn}$.

\begin{theorem}\citep{quinn2012directed}
For a joint distribution $P_{\allX}$, if $P_{\allX}$ is positive and strictly causal, then the parent sets $\{A(i)\}_{i=1}^m$ in the directed information graph are the unique, minimal cardinality parent sets such that $ \D(P_{\allX} \|  \prod_{i=1}^m  P_{\X_i \parallel \allX_{A(i)} } ) = 0.$ 
\end{theorem}

A graphical separation criterion, similar to d-separation for Bayesian networks, applies to directed information graphs \citep{eichler2012graphical}.


\section{Optimal Bounded In-Degree Approximations}\label{sec:optbnd}

When the exact topology is not necessary or is prohibitive to learn, approximations can be useful.   Approximations with simple topologies facilitate visual comprehension and in some cases can be efficient to identify.  We investigate  algorithms to identify optimal approximations for two settings.  Goodness of the approximations is measured by the KL divergence between the full joint distribution and the distribution induced by the approximation. The researcher specifies the in-degrees, controlling the complexity.  Also, the optimal approximations will be identified using low dimensional statistics, not the whole joint distribution. 

We consider approximations of the form
\beqa
\Phat_{\allX}( \allx ) := \prod_{i=1}^m P_{\bX_{i}\parallel\bX_{A(i)}} (\bx_{i}\parallel \bx_{A(i)}), \label{eq:best_parent_appx}
\eeqa
where the $A(i) \subseteq [m] \backslash\{i\}$ are candidate parent sets and the marginal distributions $\{P_{\bX_{i}\parallel\bX_{A(i)}}\}_{i=1}^m$ are exact.  Let $\calG$ denote the set of such approximations.  The goal is to find the $\Phat_{\allX} \in \calG$ that minimizes the KL divergence $\D (P_{\allX} \parallel \Phat_{\allX})$.  The following theorem characterizes an important decomposition property for evaluating the quality of an approximation $\Phat_{\allX}$.  The approximation that minimizes the KL divergence is the one that maximizes a sum of directed informations from parent sets to children.

\begin{theorem} \citep{quinn2013efficient} For any distribution $P_{\allX}$,
\label{thm:apx:gen_appx}
 \beqa
 \argmin_{\Phat_{\allX} \in  \calG}   \D ( P_{\allX} \parallel \Phat_{\allX} )  &=&  \argmax_{\Phat_{\allX} \in  \calG }  \sum_{i=1}^{m}  \I ( \X_{A(i)} \to  \X_{i} ). \label{eq:apx:gen:thm}
\eeqa
\end{theorem} 
\begin{remark}
 In \cite{quinn2013efficient}, only the specific case $|A(i)|=1$ was considered but the proof naturally extends to the general case.
\end{remark}

This decomposition property will be important for the following results.

\subsection{An Unconstrained Formulation}

Consider finding an optimal approximation of the form \eqref{eq:best_parent_appx} where the only constraint is that the in-degrees are $|A(i)|=K\geq 1$.   We assume uniform $K$ for simplicity.  The results hold if $K$ is a function of $i$.  Let $\calG_K$ denote the set of all such approximations.  The formula \eqref{eq:apx:gen:thm} simplifies.

\begin{corollary} \citep{quinn2012directed}
\label{cor:apx:gen_K_apx} For any distribution $P_{\allX}$, the parent sets $\{A^*(i)\}_{i=1}^m$ corresponding to an optimal approximation $\Phat^* \in \argmin_{\Phat_{\allX} \in  \calG_K}   \D ( P_{\allX} \parallel \Phat_{\allX} )$ satisfy
 \beqa
  A^*(i) \in  \argmax_{A(i)  : |A(i)|=K }  \I ( \X_{A(i)} \to  \X_{i} ). \nonumber 
\eeqa
\end{corollary}

Thus, finding the optimal structure is equivalent to finding the best individual parent sets for each node.  The process is described in Algorithm~1.  A modified Algorithm~1 for exact structure learning was presented in  \cite{quinn2012directed}.  Algorithm~1 takes as input the following set of directed information values, \begin{eqnarray} %
\mathcal{DI}_{\mathrm{BndInd}} &=& \left\{ \I( \allX_{ B(i)} \to \bX_i  ) : i\in [m], B(i) \subseteq \setmi{m}{i}, |B(i)|=K  \right \} . \nonumber
\end{eqnarray}

\begin{table}[t!] \begin{normalsize} \begin{center}
\begin{tabular*}{\linewidth}{@{}llrr@{}}
{\bfseries Algorithm 1.  {\sc OptimalGeneral}} \citep{quinn2012directed}\\
\hline {\bf Input:} $\mathcal{DI}_{\mathrm{BndInd}}, K, \ m$ \\
\hline

~1. {\bf For} $i \in [m]$  \\
~2. $\quad\  A(i) \gets \emptyset$ \\
~3. \(\quad \ \calB \gets \{ B:B\subseteq \setmi{m}{i}, \ |B| = K  \}\) \\
~4. $\quad $  \( A(i) \gets  \argmax_{B \in \calB} \ \I( \allX_{B} \to \bX_i)\) \\
~5. {\bf Return} $\{A(i)\}_{i=1}^m$\\
\hline
\end{tabular*}\end{center}\label{alg:StructRecov:optgen} \end{normalsize}
\end{table}

\begin{theorem} \citep{quinn2012directed} \label{thm:alg1:opt_gen}
Algorithm~1 returns an optimal approximation $\PT \in \calG_K $.
\end{theorem}  

We next consider a more specific class of graph structures.

\subsection{Finding a Connected Graph}

Algorithm~1 might return an unconnected graph.  For situations where information or influence propagates in the network, it can be better to work with connected structures.  Directed trees, the simplest connected structure, were investigated in \cite{quinn2013efficient}.  While visually simple and computationally easy to identify, they cannot depict complex dynamics such as feedback.  We next consider a balance between the properties of unconstrained bounded in-degree approximations and directed trees.  The new approximations contain a directed spanning tree as a subgraph and have user-specified in-degrees.  See Figure~\ref{fig:diag:tree2in}.  Note that the root node has no parents.  Remark~\ref{remark:rootpar} will explain how to obtain graphs where the root also has parents.

Let $\widetilde{\calG}_K$ be the set of all graphs containing a spanning tree and all nodes except the root have in-degree $K\geq 1$.  Let $\tildeA(i,j)$ be the best set of $K$ parents  for $\X_i$ that contains the edge $\X_j \to \X_i$,
\beqa
\tildeA(i,j) = \argmax_{A(i): A(i) \subseteq [m]\backslash \{i\} ,j \in A(i)} \I(\allX_{A(i)} \to \X_i). \label{eq:algs:optconAtil}
\eeqa
Then assign weight $\I( \X_{\tildeA(i,j)} \to \X_i)$ to edge $\X_j \to \X_i$ in the complete graph and run a maximum weight directed spanning tree (MWDST) algorithm. Each edge $\X_j \to \X_i$ in the spanning tree induces the corresponding parent set $\tildeA(i,j)$ for $\X_i$.  This process is described in Algorithm~2. 

\begin{table}[t!] \begin{normalsize} \begin{center}
\begin{tabular*}{\linewidth}{@{}llrr@{}}
{\bfseries Algorithm 2. {\sc OptimalConnected}}\\
\hline {\bf Input:} $\mathcal{DI}_{\mathrm{BndInd}}, K, \ m$ \\
\hline

~1. {\bf For} $i \in [m]$  \\
~2. $\quad \  A(i) \gets \emptyset$ \\
~3. $\quad$  {\bf For} $j \in [m]\backslash \{i\}$  \\
~4. \(\quad \quad \ \ \calB \gets \{ B:B\subseteq \setmi{m}{i}, \ |B| = K, \ j \in B  \}\) \\
~5. \(\quad \quad \ \ \tildeA(i,j) \gets  \argmax_{B \in \calB} \ \I( \allX_{B} \to \bX_i)\) \\
~6. $\{a(i)\}_{i=1}^m \gets $ {\bf MWDST} $( \{ \I( \allX_{\tildeA(i,j)} \to \X_i) \}_{1\leq i\neq j \leq m})$ \\
~7. {\bf For} $i \in [m]$  \\
~8. $\quad\  A(i) \gets \tildeA(i,a(i))$ \\
~9. {\bf Return} $\{A(i)\}_{i=1}^m$\\
\hline
\end{tabular*}\end{center}\label{alg:StructRecov:optcon} \end{normalsize}
\end{table}

\begin{theorem} \label{thm:alg2:opt_con}
Algorithm~2 returns an optimal approximation $\PT \in \widetilde{\calG}_K $.
\end{theorem}  

The proof is in Appendix~\ref{apdx:thm:alg2:opt_con}.

\begin{remark} \label{remark:rootpar}
The approximations $\PT \in \widetilde{\calG}_K $ have root nodes with no inward edges.  Algorithm~2 can be modified to find the best approximation where all nodes have in-degree $K$ and there is a directed spanning tree as a subgraph.  Namely, create a dummy node $\X_0$, set edge weights $\I(\X_j \to \X_0) \gets - \infty$ and $\I(\X_0 \to \X_j) \gets - 1$ for all $j\in [m]$.  Note that all the other edge weights are directed informations, which are KL divergences and hence non-negative.  Then Algorithm~2 will set $\X_0$ as the root with a single outward edge. 
\end{remark}

Algorithms~1~and~2 find optimal approximations in terms of KL divergence $\D( P_{\allX} \| \PT)$.  They only need distributions over $K+1$ processes, not the full joint distribution.  However, they compute  $m{ m-1 \choose K}$ directed informations involving $K$ processes.  If $K$ is large, this could be computationally difficult.  For some applications, instead of reducing $K$, it is better to efficiently identify near-optimal approximations.

\section{Near-Optimal Bounded In-Degree Approximations}\label{sec:eff_near_opt_con}

We next find sufficient conditions to identify near-optimal approximations in time polynomial in $K$.

\subsection{Greedy Submodularity}

Consider the following greedy procedure to select a parent set for $\X_i$. Initially, set $\X_i$'s parent set as the best individual parent $\Z = {\argmax}_j \hspace{0.1cm} \I(\X_j \to \X_i)$.  Then look for the second best parent $\Z'={\argmax}_j \hspace{0.1cm} \I(\X_j \to \X_i \| \Z).$   Repeat this $K-2$ times, adding one parent at each iteration. 

In general, greedy methods are not provably good.  We next describe sufficient conditions to guarantee near-optimality.

\begin{definition} \label{def:grdsub} A joint distribution $P_{\allX}$ is called \emph{greedily-submodular} if there exists an $\alpha>0$, such that for any process $\Y$ and any subset $\allX_{\uW}$ of other processes, \beqa  \I(\X_{j} \to \Y \| \X_1, \dots, \X_{j-2}, \X_{j-1})  \leq \alpha  \I(\X_{j-1} \to \Y \| \X_1, \dots, \X_{j-2}), \label{eq:def:grdsub} \eeqa for all $1\leq j < |\uW|$ where the processes in $\allX_{\uW}$ are indexed according to the order in which they are selected by the greedy algorithm.
\end{definition}

This is a weaker condition than submodularity, a discrete analog of concavity \citep{nemhauser1978analysis}.  If $P_{\allX}$ has submodular directed information values, then for all pairs of processes $\{\X_j, \Y\}$ and sets of other processes $\allX_{S} \subseteq \allX_{S'} \subseteq \allX \backslash \{\X_j,\Y \}$, 
\beqa \I (\X_j \to \Y \| \allX_{S'}) \leq \I (\X_j \to \Y \| \allX_{S}). \label{eq:submod:def} 
\eeqa  
Submodularity implies conditioning does not increase directed information.  %
\begin{corollary} If $P_{\allX}$ is submodular, it is also greedily-submodular with $\alpha \leq 1$. 
\end{corollary}
\begin{proof}
Let $S = \{1,\ldots,j-2\}$ and $S' = S\cup \{j-1\}$, and let the processes be labeled in the order they are selected in a greedy search for parents for $\Y$.  Then if $P_{\allX}$ is submodular, 
\begin{align}  \I(\X_{j} \to \Y \| \X_1, \dots, \X_{j-2}, \X_{j-1})  
	&\leq    \I(\X_{j} \to \Y \| \X_1, \dots, \X_{j-2}),\label{eq:submod:a} \\
	&\leq   \I(\X_{j-1} \to \Y \| \X_1, \dots, \X_{j-2}),\label{eq:submod:b}  
\end{align} where \eqref{eq:submod:a} would hold by \eqref{eq:submod:def} and \eqref{eq:submod:b} would hold because $\X_j$ is picked after $\X_{j-1}$ in a greedy search.  Thus, \eqref{eq:submod:def} holds with $\alpha \leq 1$.
\end{proof}

Entropy is submodular \citep{fujishige1978polymatroidal}. However, in general mutual information and directed information are not, as shown in the following example.

\begin{example} \label{ex:DInotsubmod} Let $\{\N,\X,\Z\}$ be mutually independent, zero-mean, i.i.d.\ Gaussian processes. Let $Y_{t+1} = X_t + Z_t + N_t$.  Then using stationarity \cite[pg.~256]{cover2006elements},
 \beqas \I(\X \to \Y) 
 = \I(X_1; Y_2) 
 &=& \frac{1}{2} \log \left( 1 + \frac{\var (X_1)}{\var (Z_1)+\var(N_1)} \right) \\
 &<& \frac{1}{2} \log \left( 1 + \frac{\var (X_1)}{\var(N_1)} \right) \\ 
 &=& \I(X_1; Y_2 | Z_1) 
 =\I(\X \to \Y \| \Z).  
 \eeqas  Since conditioning can {\em increase} directed information, it is not submodular.  
\end{example}

\begin{remark} The authors are not aware of this property being discussed in the literature previously.  Two other conditions that are weaker than submodularity are discussed in \cite{cevher2011greedy} and \cite{das2011submodular}.  The former uses submodularity up to an additive error.  The latter uses submodularity up to multiplicative error. Both measure the increase in conditioning of the terms in \eqref{eq:submod:def}, unlike \eqref{eq:def:grdsub} which only bounds sequential increases while greedily selecting a parent set.
\end{remark}

\begin{assumption} \label{assump:greedysub} We assume that $P_{\allX}$ is greedily-submodular.
\end{assumption}

When Assumption~\ref{assump:greedysub} holds, the greedy search yields a near-optimal approximation.  Let $A$ denote the set of indices for an optimal set of $K$ parents and $B$ the indices for the greedily selected set of $L\leq K$ parents. 

\begin{theorem} \label{thm:grd_bnd} Under Assumption~\ref{assump:greedysub}, \[ \I(\allX_{B} \to \Y) \geq \left(  1 - \exp \left( \frac{-L}{\sum_{i = 0}^{K-1} \alpha^i} \right)  \right) \I(\allX_{A} \to \Y). \]
\end{theorem}  The proof is in Appendix~\ref{apnd:prf:grd_bnd}.

Recall from Theorem~\ref{thm:apx:gen_appx} that the larger the sum of directed information values from parent sets to children is, the better the approximation is.  Theorem~\ref{thm:grd_bnd} implies that greedy approximations are near-optimal approximations.  Figure~\ref{fig:bnd_vals} shows the bound coefficient in Theorem~\ref{thm:grd_bnd} for $\alpha \in \{ 1.3, 1.7, 2.5\}$. In Example~\ref{ex:DInotsubmod}, if the variances were equal, $\alpha = 1.71$ would suffice.

\newcommand{\Aopt} {\uA_{\text{opt}}}
\newcommand{\Agrd} {\uA_{\text{grd}}}

We can also bound how close an optimal parent set $A_L$ with in-degree $L$ is to an optimal parent set $A_K$  with in-degree $K>L$.  
\begin{corollary} \label{cor:KbndL} Under Assumption~\ref{assump:greedysub}, with $\alpha \neq 1$, 
\[  \I(\allX_{A_L} \to \Y) \geq  \left( \frac{\alpha^{L} - 1}{\alpha^{K} - 1} \right)  \I(\allX_{A_{K}} \to \Y).\] 
\end{corollary}   The proof is in Appendix~\ref{apnd:prf:cor:KbndL}.  The bound coefficient is plotted in Figure~\ref{fig:bnd_vals_LK}.

\begin{figure}[t]
\centering
  \subfigure[The bound in Theorem~\ref{thm:grd_bnd} with $L=K$.
  ]{\label{fig:bnd_vals} \includegraphics[width=.45\columnwidth]{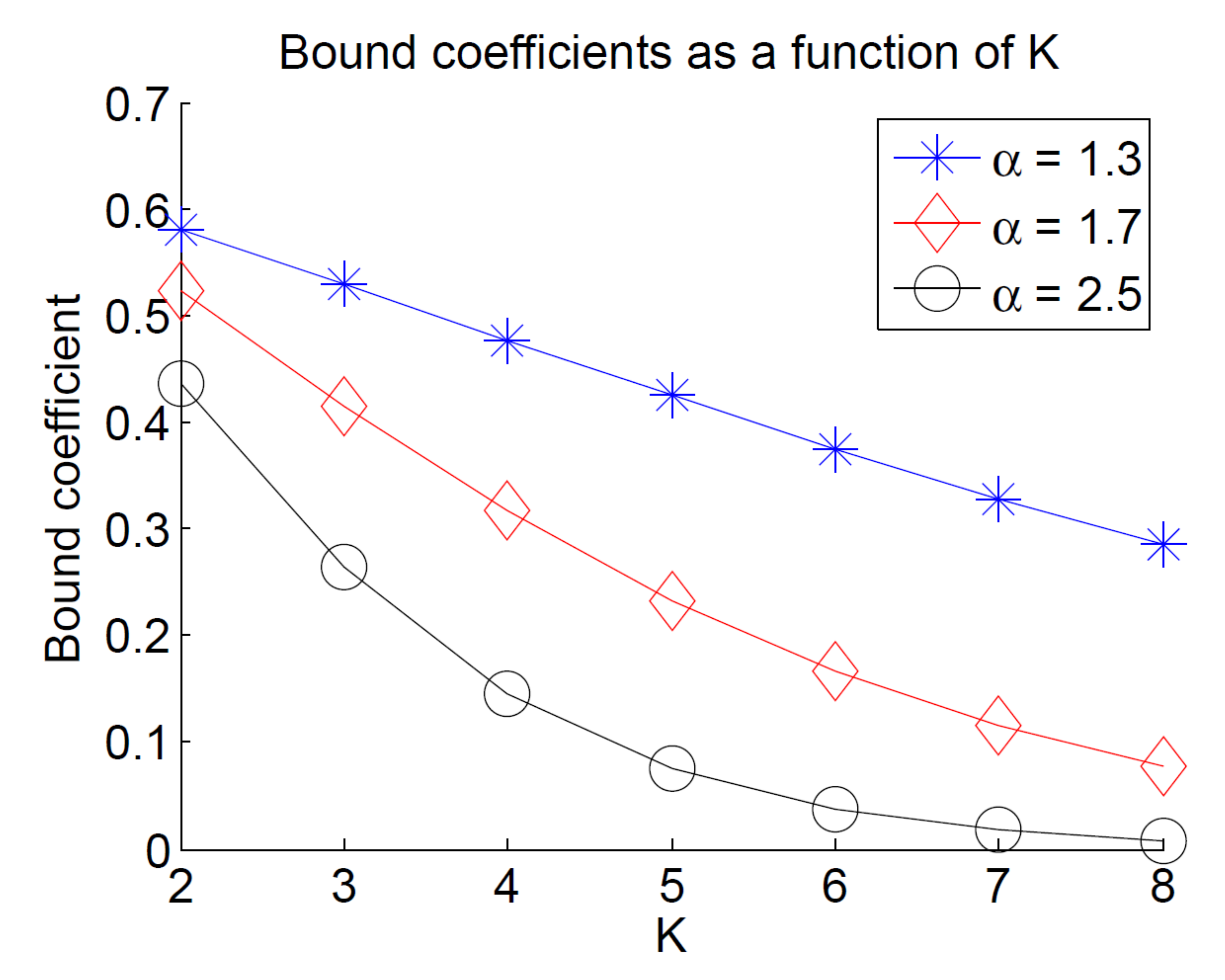}} \hspace{0.3cm}
  \subfigure[The bound in Corollary~\ref{cor:KbndL} with $L=2$.
  ]{\label{fig:bnd_vals_LK} \includegraphics[width=.45\columnwidth]{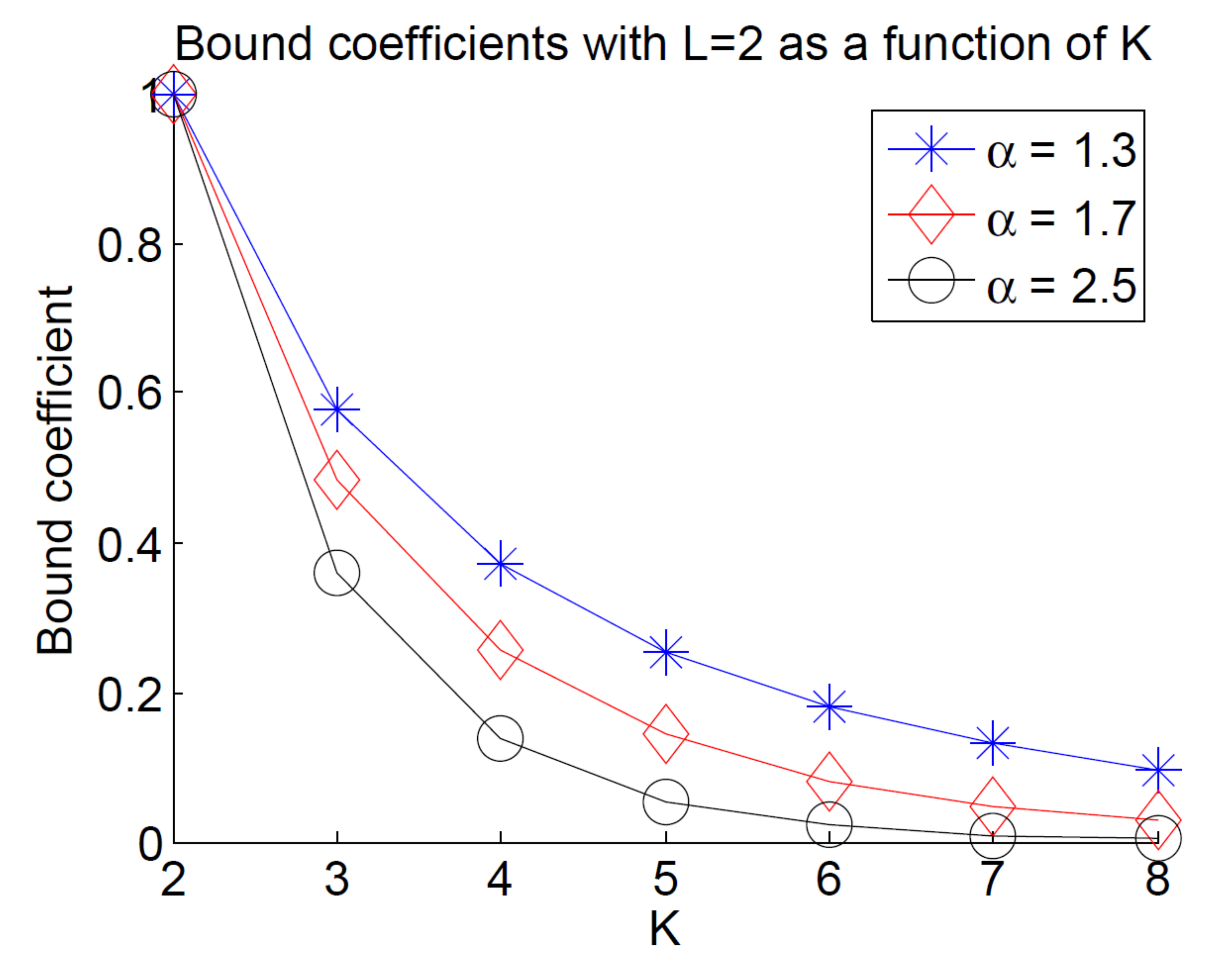}}  
  \caption{\small Plots for the bounds in Theorem~\ref{thm:grd_bnd} and Corollary~\ref{cor:KbndL} respectively.  }
  \label{fig:bnds}
\end{figure}

\subsection{Near-Optimal Solutions for  the  Unconstrained Problem}

We next consider Algorithm~3 which uses a greedy search to find a near-optimal solution to \eqref{eq:apx:gen:thm} where the only constraints are in-degree bounds, similar to Algorithm~1.   Since the greedy search is adaptive, directed information values will be computed as needed.   Let $\{B(i)\}_{i=1}^m$ and $\{A(i)\}_{i=1}^m$  denote the parent sets returned by Algorithms 3 and 1 respectively.

\begin{table}[t] \begin{normalsize} \begin{center}
\begin{tabular*}{\linewidth}{@{}llrr@{}}
{\bfseries Algorithm 3. {\sc Near-OptimalGeneral}}\\
\hline {\bf Input:} $L$, $m$ \\
\hline
~1. {\bf For} $i \in \{1, \dots, m\}$  \\
~2. \quad $B(i) \gets \emptyset$ \\
~3. \quad {\bf While} $|B(i)| < L$ \\
~4. \quad \quad {\bf For} $l \in \{ 1, \dots, m\} \backslash \{B(i) \bigcup \{i\}  \}$ \\
~5. \quad \quad \quad {\bf Compute} $\I( \X_l \to \X_i \| \allX_{B(i)})$\\
~6. \quad \quad  $B(i)  \gets  B(i) \bigcup \argmax_{l } \hspace{0.02cm} \I( \X_l \to \X_i \| \allX_{B(i)})$ \\
~7. {\bf Return} $\{B(i)\}_{i=1}^m$\\
\hline
\end{tabular*}\end{center}\label{alg:caus_dep_tree2:optgen} \end{normalsize}
\end{table}

\begin{theorem}\label{thm:grd_gen} Under Assumption~\ref{assump:greedysub}, 
\beqas 
 \sum_{i=1}^m \I(\allX_{B(i)} \to \X_i) \geq \left( \hspace{-0.0cm} 1 - \exp\left( \frac{-L}{\sum_{i = 0}^{K-1} \alpha^i} \right)  \right)  \hspace{-0.0cm} \sum_{i=1}^m \I(\allX_{A(i)}  \to \X_i) .
\eeqas
\end{theorem}  \begin{proof} The proof follows from Theorem~\ref{thm:grd_bnd} holding for each $i \in \{1,\dots, m\}$. \end{proof} %
\begin{remark}
The edge weight $\I( \allX_{B(i)} \to \X_i)$ can be computed using a chain rule \citep{kramer1998directed}. Let $\{j_1, j_2, \dots, j_{L} \}$ denote $B(i)$.  Then \[ \I( \allX_{B(i)} \to \X_i) = \sum_{l = 1}^{L}  \I ( \X_{j_l} \to \X_i \| \X_{j_1},  \dots \X_{j_{l-1}}) .\]
\end{remark}

\subsection{Near-Optimal Solutions for  Finding Connected Graphs}

We now propose Algorithm~4 which uses a greedy search to find a near-optimal connected solution to \eqref{eq:apx:gen:thm}.  Similar to Algorithm~3, it precomputes parent sets for each possible directed edge.  Then a MWDST algorithm is called.   In Algorithm~4, $\widetilde{B}(i,j)$ is the set of parents for $\X_i$ with $\X_j$ as one of the parents, selected in a greedy fashion. The value $\I( \allX_{\wtB(i,j)} \to \X_i)$ is the weight of edge $\X_j \to \X_i$ given to the MWDST algorithm.  Let $\{A(i)\}_{i=1}^m$ denote the parent sets returned by Algorithm~2.

\begin{table}[t!] \begin{normalsize} \begin{center}
\begin{tabular*}{\linewidth}{@{}llrr@{}}
{\bfseries Algorithm 4. {\sc Near-OptimalConnected}}\\
\hline {\bf Input:} $L$, $m$ \\
\hline
~\phantom{0}1. {\bf For} $i \in \{1, \dots, m\}$  \\
~\phantom{0}2. \quad $B(i) \gets \emptyset$ \\
~\phantom{0}3. \quad  {\bf For} $j \in \{1, \dots, m\}\backslash \{i\}$ \\
~\phantom{0}4. \quad \quad $\wtB(i,j) \gets \{j\}$\\
~\phantom{0}5. \quad \quad {\bf While} $|\wtB(i,j)| <L$ \\
~\phantom{0}6. \quad \quad  \quad {\bf For} $l \in \{ 1, \dots, m\} \backslash \{\wtB(i,j) \bigcup \{i\}  \}$ \\
~\phantom{0}7. \quad \quad  \quad \quad {\bf Compute} $\I( \X_l \to \X_i \| \allX_{\wtB(i,j)})$\\
~\phantom{0}8. \quad \quad  \quad  $\wtB(i,j) \!\! \gets \!\!\wtB(i,j) 
\bigcup \argmax_{l } \hspace{0.02cm} \I( \X_l \to \X_i \| \allX_{\wtB(i,j)})$ \\
~\phantom{0}9. $\{b(i)\}_{i=1}^m \gets $ {\bf MWDST} $( \{ \I( \allX_{\wtB(i,j)} \to \X_i) \}_{1\leq i\neq j \leq m})$ \\
~10. {\bf For} $i \in \{1, \dots, m\}$  \\
~11. $\quad\  B(i) \gets \wtB(i,b(i))$ \\
~12. {\bf Return} $\{B(i)\}_{i=1}^m$\\
\hline
\end{tabular*}\end{center}\label{alg:caus_dep_tree2:optcon} \end{normalsize}
\end{table}

\begin{theorem}\label{thm:grd_connect} Under Assumption~\ref{assump:greedysub}, for Algorithm~4, 
\beqas 
 \sum_{i=1}^m \I(\allX_{B(i)} \to \X_i) \geq \left( \hspace{-0.0cm} 1 - \exp\left( \frac{-L}{\sum_{i = 0}^{K-1} \alpha^i} \right)  \right)  \hspace{-0.0cm} \sum_{i=1}^m \I(\allX_{A(i)}  \to \X_i) .
\eeqas
\end{theorem}   The proof is in Appendix~\ref{apnd:prf:grd_connect}.

\section{Best $r$ Bounded In-Degree Approximations} \label{sec:best_r}

For a given $P_{\allX}$, Algorithms~1--4 each return a single solution.  For many applications, knowing the $r$-best approximations, where $r$ might be five, twenty, or a hundred, can be advantageous.  For instance, there is no guarantee of uniqueness for optimal approximations.  Additionally, when data is limited or noisy, the actual optimal solution might appear as second or third best due to estimation errors.  Also, for more complex constraints, such as directed trees with depth at most five, it might be easier to find the $r$-best solutions to the more general problem of directed tree approximations, and then pick the highest ranking solution that satisfies the extra constraint.  Lastly, edges that persist in all of the $r$-best approximations are likely to be important.

We next discuss methods to identify the $r$-best bounded in-degree approximations.  For simplicity, we will focus on altering Algorithm~1 and then discuss differences for modifying the other algorithms.  A strategy for identifying the $r$-best solutions for assignment problems is discussed in \cite{lawler1972procedure} based on branching candidate solutions.  The method would be impractical to apply here.   We develop an alternative algorithm, which, like that in \cite{lawler1972procedure}, will be based on a branching of candidate solutions.

\subsection{Optimal Solutions for the Unconstrained Problem} \label{sec:rbest:optgen}

Recall from Theorem~\ref{thm:apx:gen_appx} that the sum of directed information values from parent sets to children corresponds to how good an approximation is.  When the only constraint is a user specified in-degree $K$, the parent sets can be chosen independently, as in Algorithm~1.  This property simplifies searching for the $r$-best approximations.  For instance, the second best approximation will only differ from the first by one parent set.  

Algorithm~5 identifies the $r$-best bounded in-degree approximations in order.  It maintains a list of candidate approximations.  It instantiates that list by calling Algorithm~1.  Each time an approximation is selected from the list, Algorithm~6 generates new candidate approximations from that ``seed.''  Algorithm~6 finds $m$ solutions by keeping all but one parent set, replacing it with the next best one.

To simplify the presentation, we will make the following assumption to avoid ties.
\begin{assumption} \label{assump:nonequalpar} For a given joint distribution $P_{\allX}$,  for any process $\X_i$, no two parent sets $A(i), B(i) \subseteq [m] \backslash \{i\}$, with $|A(i)| = |B(i)|$ have identical directed information values,
$\I(\allX_{A(i)} \to \X_i) \neq \I(\allX_{B(i)} \to \X_i).
$
\end{assumption}

For cases where Assumption~\ref{assump:nonequalpar} does not hold, line~4 in Algorithm~6 can be modified to check not only values, but elements of parent sets as well.

\begin{table}[t!] \begin{normalsize} \begin{center}
\begin{tabular*}{\linewidth}{@{}llrr@{}}
{\bfseries Algorithm 5. {\sc TopRGeneral}}\\
\hline {\bf Input:} $\mathcal{DI}_{\mathrm{BndInd}}, K, m, r$ \\
\hline
~1. $Top \gets \emptyset$  \\
~2. $l \gets 0$  \\
~3. $\calS \gets \text{\sc OptimalGeneral}(\mathcal{DI}_{\mathrm{BndInd}}, K, \ m)$  \\
~4. {\bf While} $l < r$  \\
~5. $\quad\  l  \gets l+1$ \\
~6. $\quad\  Top(l) \gets \argmax_{\Phat_{\allX} \in \calS} \sum_{i=1}^m \I(\allX_{A(i)} \to \X_i)$ \\
~7. $\quad\  \calS \gets \calS \bigcup \text{\sc GetNewSoln}(\mathcal{DI}_{\mathrm{BndInd}}, K, m,Top(l))$ \\
~8. $\quad\  \calS \gets \calS \ \backslash \ Top$ \\
~9. {\bf Return} $Top$\\
\hline
\end{tabular*}\end{center}\label{alg:TopRGeneral} \end{normalsize}
\end{table}

\begin{table}[t!] \begin{normalsize} \begin{center}
\begin{tabular*}{\linewidth}{@{}llrr@{}}
{\bfseries Algorithm 6. {\sc GetNewSolns}}\\
\hline {\bf Input:} $\mathcal{DI}_{\mathrm{BndInd}}, K,  m, \{A'(i)\}_{i=1}^m$ \\
\hline

~1. $\calS \gets \emptyset$  \\
~2. {\bf For} $i \in \{1,\dots,m\}$  \\
~3. $\quad\ \{A(i)\}_{i=1}^m \gets \{A'(i)\}_{i=1}^m$ \\
~4. \(\quad \ \calB \gets  \{ B:B\subseteq \setmi{m}{i}, \ |B| = K, \  \I(\allX_B \to \X_i) < \I(\allX_{A'(i)} \to \X_i)   \} \) \\
~5. $\quad $  \( A(i) \gets  \argmax_{B \in \calB} \ \I( \allX_{B} \to \bX_i)\) \\
~6. $\quad\  \calS \gets \calS \bigcup \{ \{A(i)\}_{i=1}^m \}$ \\
~7. {\bf Return} $\calS$\\
\hline
\end{tabular*}\end{center}\label{alg:GetNewSoln} \end{normalsize}
\end{table}

\begin{theorem}
\label{thm:apx:TopRGeneral} Under Assumption~\ref{assump:nonequalpar}, Algorithm~5 returns the $r$-best bounded in-degree approximations.
\end{theorem}  

The proof is in Appendix~\ref{apnd:prf:TopRGeneral}.  We also provide a discussion for how to index approximations for efficient implementation of lines~6--8 of Algorithm~5 in Appendix~\ref{apnd:disc:TopRGeneral}.

\subsection{Optimal Solutions to Find Connected Graphs}\label{sec:rbest:optcon}

Identifying the $r$-best {\em connected} approximations is more complicated than the general case in Section~\ref{sec:rbest:optgen}.  Algorithm~5 would not need to be modified, but Algorithm~6 would.  Recall from Algorithm~2 line~5 that each edge $\X_j\to \X_i$ in the complete graph is assigned an edge weight $\I(\allX_{\tildeA(i,j)} \to \X_i)$.  If the edge $\X_j \to \X_i$ is selected by the MWDST algorithm in line~6, then $\tildeA(i,j)$ is the parent set assigned to $\X_i$.  To find the $r$-best connected approximations, as in Algorithm~6,  ``seed'' approximations should be modified to generate new candidate approximations.  However, Algorithm~6 will not work properly.  Some candidate approximations might be identical to the seed.

To see this, let $\{j_1, j_2, \dots, j_K\}$ denote $\tildeA(i,j_1)$.  Suppose $\tildeA(i,j_1) = \tildeA(i,j_2)$ and $\X_{j_1} \to \X_i$ was an edge in the MWDST that induced the seed approximation.  Thus, $\tildeA(i,j_1)$ is the parent set selected for $\X_i$.  In generating new approximations, even if $\X_{j_1} \to \X_i$ is given a smaller weight, $\I(\allX_{\tildeA'(i,j_1)} \to \X_i)$, edge $\X_{j_2} \to \X_i$ might be selected by the MWDST algorithm instead, yielding the same parent set.

One approach to generate candidate approximations involves checking whether modifying a single edge weight in the complete graph, such as setting $\X_{j_1} \to \X_i$ to have weight $\I(\allX_{\tildeA'(i,j)} \to \X_i)$, does result in a candidate approximation different than the seed.  If not, then all subsets of edges inducing the same parent set $\tildeA(i,j_1)$ should be modified.  Thus if $\tildeA(i,j_1) = \tildeA(i,j_2) = \tildeA(i,j_3)$, then the weights of edges $\{\X_{j_1}\to \X_i\}$, $\{\X_{j_1}\to \X_i, \X_{j_2}\to \X_i$\} $\{\X_{j_1}\to \X_i, \X_{j_3}\to \X_i\}$, and $\{\X_{j_1}\to \X_i, \X_{j_2}\to \X_i, \X_{j_3}\to \X_i\}$ should be modified.  Any resulting candidate parent sets that differ from the seed should be retained.

\subsection{Near-Optimal Solutions for the Unconstrained Problem}

Algorithm~5 generates the $r$-best approximations, but calls Algorithms~6 which uses an exhaustive search.  A greedy search can be used instead to generate $r$ approximations.  Consider the first time that Algorithm~6 is called.  Let $\{ j_1, j_2, \dots, j_K \}$ denote the parent set, in order they were added, for node $i$.  When $i$'s parent set is changed, in line~5 set \[ \calB \gets \{ B: \{ j_1, \dots, j_{K-1} \} \subseteq B\subseteq \setmi{m}{i, j_K}, \ |B| = K  \}. \]  So only the parent added last in the greedy search is changed.   

Consider the first time that Algorithm~6 is called where $i$'s parents in the seed approximation is $\{ j_1, \dots, j_{K-1}, j_K' \}$ for some $j_K' \neq j_K$.  Then set 
\beqas  \calB \gets \{ B: \{ j_1, \dots, j_{K-1} \} \subseteq B\subseteq \setmi{m}{i, j_K, j_K'},  \ |B| = K  \}. 
\eeqas  This can be repeated $m-K-2$ more times until $|\calB| = \emptyset$.  When this occurs, the $ (K-1)$th parent needs to be changed.  Set 
\beqas  \calB \gets \{ B: \{ j_1, \dots, j_{K-2},j'_{K-1} \} \subseteq B\subseteq \setmi{m}{i,j_{K-1}},  \ |B| = K  \}. 
\eeqas where \[j'_{K-1} = \argmax_{j \in [m]\backslash \{i,j_{K-1} \}} \I( \X_{j} \to \X_i \| \allX_{j_1, \dots, j_{K-2}}),\] the next $(K-1)$th parent selected in a greedy order.  Continue in this manner until Algorithm~5 selects $r$ approximations.

Note also that we can combine the modifications discussed here with those in Section~\ref{sec:rbest:optcon} to identify the top $r$ connected approximations using a greedy search.

\section{Complexity of Proposed Algorithms} \label{sec:complexity}

This section explores the computational complexity of the algorithms and storage complexity of the approximations.  

First, calculating $\I(\X, \Z \to \Y)$ in general has exponential complexity.  Note that \[ \I(\X, \Z \to \Y) = \sum_{t=1}^n \I(Y_t ; X^{t-1}, Z^{t-1} | Y^{t-1}). \]  The last term in particular, $\I(Y_n ; X^{n-1}, Z^{n-1} | Y^{n-1})$, involves a sum over all realizations of $3n - 2$ random variables.  Thus, with $\calX$ denoting the alphabet, the last term has complexity $\calO( |\calX|^{3n})$.  We will assume Markovicity of a fixed order $l$, so\[ \I(\X, \Z \to \Y) = \sum_{t=1}^n \I(Y_t ; X_{t-l}^{t-1}, Z_{t-l}^{t-1} | Y_{t-l}^{t-1}).\]  The complexity of computing $\I(\X, \Z \to \Y)$ then becomes $\calO( n |\calX|^{3l+1})= \calO( n )$.  More generally, computing $\I(\allX_{B} \to \Y \| \allX_{B'})$, where $|B| + |B'| = K$, has $\calO( n |\calX|^{(K+1)l+1})$ complexity assuming Markovicity.

\begin{assumption} \label{assump:complexity} We assume Markovicity of order $l$.
\end{assumption}

\subsection{Algorithm~1. {\sc OptimalGeneral}} For each process $\X_i$, for each of the $m-1 \choose K $ possible subsets $B$ with $|B| = K$, $\I (\allX_{B} \to \X_i)$ is computed.  Each computation has complexity $\calO( n |\calX|^{(K+1)l+1})$.  Thus, the total complexity for Algorithm~1 under Assumption~\ref{assump:complexity} is $\calO( m {m-1 \choose K } n |\calX|^{(K+1)l+1})$, or $\calO( m^{K+1} n)$ for fixed $K$.

\subsection{Algorithm~2. {\sc OptimalConnected}} Algorithm~2 computes the same directed information terms as Algorithm~1.   It also computes a MWDST, which takes $\mathcal{O}(m^2)$ time \citep{edmonds1967optimum}.  
Under Assumption~\ref{assump:complexity}, the total complexity is $\calO( m {m-1 \choose K } n |\calX|^{(K+1)l+1}+m^2)$.  If $K$ is fixed, the complexity becomes $\calO( m^{K+1} n)$.

\subsection{Algorithm~3. {\sc Near-OptimalGeneral}}
For each process $\X_i$ there are $(m-1)$ directed information terms computed involving two processes, of the form $\I(\X_j \to \X_i)$.  Next there are $(m-2)$ computed involving three processes, and so on.  The complexity is thus
\beqas 
 \calO ( m \sum_{i = 1}^{K} (m-1)n |\calX|^{(i+1)l+1}  )  = \calO (  m^2 K n |\calX|^{(K+1)l+1}).\eeqas 
For constant $K$, this becomes $\calO \left( m^2 n \right)$.

\subsection{Algorithm~4. {\sc Near-OptimalConnected}} For each ordered pair of processes $(\X_i,\X_j)$, first there are $(m-2)$ terms computed involving three processes, such as $\I(\X_k \to \X_i \| \X_j)$.  Next there are $(m-3)$ computed involving four processes, and so on.  Then a MWDST algorithm is called.  The complexity is thus
\beqas \calO ( m^2+ m (m-1)  \sum_{i = 1}^{K-1} (m-1-i)n |\calX|^{(i+2)l+1}  )   = \calO ( K m^3 n |\calX|^{(K+1)l+1}).
 \eeqas 
For constant $K$, this becomes $\calO \left( m^3 n \right)$.

\subsection{Algorithms~5. {\sc TopRGeneral}}

There are three main bottlenecks in generating the top-$r$ solutions.  The first is computing the directed information terms, the same as used in Algorithms~1~and~2, $\calO(m^{K+1} n) $ for fixed $K$.  The second is sorting those values for each process $\X_i$.  Merge sort, for example, can sort an array of $h$ elements in $\calO(h \log h)$ time \citep{katajainen1996practical}.  For each of the $m$ processes, there are ${m-1 \choose K}$ values, so sorting takes $\calO(m {m-1 \choose K} \log {m-1 \choose K}) = \calO(m^{K+1} \log m)$ for fixed $K$.  The third bottleneck is the search for candidate solutions.  Algorithm~6 generates $m$ new solutions each time it is called, replacing one parent set for each approximation.  The branching overall generates $\calO(rm)$ candidates.  For small $r$, such as $r = \calO(\log {m-1 \choose K}^m)= \calO(m\log m^K) = \calO(m\log m)$ for fixed $K$, the computing and sorting the directed information values dominates.  Recall that ${m-1 \choose K}^m$ is the total number of bounded in-degree approximations.   For large $r$, such as $r = {m-1 \choose K}^m/c$ for some $c>1$, then $r = \calO(m^{Km})$ and so the branching dominates.  The total complexity with fixed $K$ is $\calO(m^{K+1} (n +  \log m) + rm)$.

\subsection{Storage Complexity}  

An important benefit of using approximations  is that they require substantially less storage than the full joint distribution.  The full joint distribution has $mn$ random variables, and so requires $\mathcal{O}( |\calX|^{mn} )$ storage.  Under Assumption~\ref{assump:complexity}, the storage complexity of the full joint distribution is $\mathcal{O}(n |\calX|^{m(l+1)} )$ and of approximations of the form  \eqref{eq:best_parent_appx} is $\mathcal{O}(mn |\calX|^{(K+1)l+1} ) = \mathcal{O}(mn )$ for fixed $K$.

\section{Simulations} \label{sec:sims}

We investigated the performances of Algorithms~1-5 using simulated networks.  Comparisons of greedy and optimal search, unconstrained and connected approximations, and the top-$r$ approximations were studied.
\subsection{Greedy vs. Optimal Search}
We first compare the near-optimal and optimal approximations identified by  Algorithm~3 and Algorithm~1 respectively.  

\subsubsection{Setup} \label{sec:sim:grdopt:setup}
Markov order-1 autoregressive (AR) networks, of the form $\allX_t = C \allX_{t-1} + N_t$, were simulated for a given $m$ by $m$ coefficient matrix $C$ and i.i.d.\ noise vector $N_t$.   Two network sizes $m \in \{ 6, 15\}$ were tested.  For each $m$, there were $250$ trials.  In each trial, the coefficient matrix $C$ was randomly generated.  Edges (non-zero off-diagonal entries in $C$) were selected i.i.d.\ with probability $1/2$.  Non-zero AR coefficients were drawn i.i.d.\ from a standard normal distribution.  $C$ was then scaled to be stationary.  The noise process $\{N_t\}_{t=1}^n$ had i.i.d.\ entries drawn from a normal distribution with mean zero and variance $1/4$.  Data was generated for $n = 1000$ time-steps.

For each network, unconstrained bounded in-degree approximations $\{A_{\mathrm{OPT}}(i)\}_{i=1}^m$ and $\{A_{\mathrm{GRD}}(i)\}_{i=1}^m$ were computed using Algorithms~1~and~3 respectively.  For $m=6$ and $m=15$, in-degrees $K=2$ and $K=4$ were used respectively.  Performance was measured by the ratio \beqa
\frac{\sum_{i=1}^m \I(\allX_{A_{\mathrm{GRD}}(i)} \to \X_i) }{\sum_{i=1}^m  \I(\allX_{A_{\mathrm{OPT}}(i)} \to \X_i)}. \label{eq:ratio_grd_opt}
\eeqa  The value of each sum corresponds to how good that approximation is.  For  \eqref{eq:ratio_grd_opt}, the directed information values were calculated exactly using approximated parent sets.

Both algorithms computed directed information estimates using the simulated data.  The estimate for a directed information of the form $\I(\X \to \Y \| \Z)$ was computed as follows.  Least square estimates for the coefficients in two AR models, 
\begin{align}
Y_t &= b_1 Y_{t-1} + b_2 Z_{t-1} + b_3 X_{t-1} + N_t \\
Y_t &= b_1' Y_{t-1} + b_2' Z_{t-1}  + N_t', 
\end{align} were computed.  Let $\sigma$ and $\sigma'$ denote $\std(N_t)$ and $\std(N_t')$ respectively. 
The entropy $\H(Y_t | Y_{t-1}, Z_{t-1}, X_{t-1})$ is $1/2 \log_2( 2 \pi e \sigma^2)$ \citep[Theorem~8.4.1]{cover2006elements}, so
\begin{align}
\widehat{\I}(\X \to \Y \| \Z) &= \frac{1}{n} \sum_{t=1}^n \widehat{\H}(Y_t | Y_{t-1},  Z_{t-1})  -\widehat{\H}(Y_t | Y_{t-1}, Z_{t-1}, X_{t-1})    \nonumber\\
&=\frac{1}{2} \log( 2 \pi e (\sigma')^2) - \frac{1}{2} \log( 2 \pi e \sigma^2) \nonumber\\
&=\log \sigma'/ \sigma. \nonumber
\end{align}

\subsubsection{Results} \label{sec:sim:grdopt:results} The approximations found by the greedy and optimal search were largely identical.  Figure~\ref{fig:greedy_optim} shows histograms of percentages of the ratio \eqref{eq:ratio_grd_opt}, normalized by the number of trials.  The rightmost column in each histogram corresponds to  \eqref{eq:ratio_grd_opt} being one, when the greedy search returned the optimal approximation.  

For $m = 6$, the greedy search found the optimal solution in  96.4\% of the trials.  The average ratio was $99.9\%\pm0.6\%$.  The minimum ratio was 92.0\%.

For $m = 15$, the greedy search found the optimal solution in  57.6\% of the trials.  The average ratio was $99.6\%\pm0.9\%.$  The minimum ratio was 93.3\%.

By Theorem~\ref{thm:grd_bnd}, the best case lower bound of \eqref{eq:ratio_grd_opt} expected is $63.2\%$, which corresponds to $\alpha=1$.  On average, the greedy algorithm performed much better than the lower bound, often the same as or close to the optimal.

\begin{figure}[t]
\centering
  \subfigure[$m = 6$.]{\label{fig:greedy_optim:m_6_K_2} \includegraphics[width=.40\columnwidth]{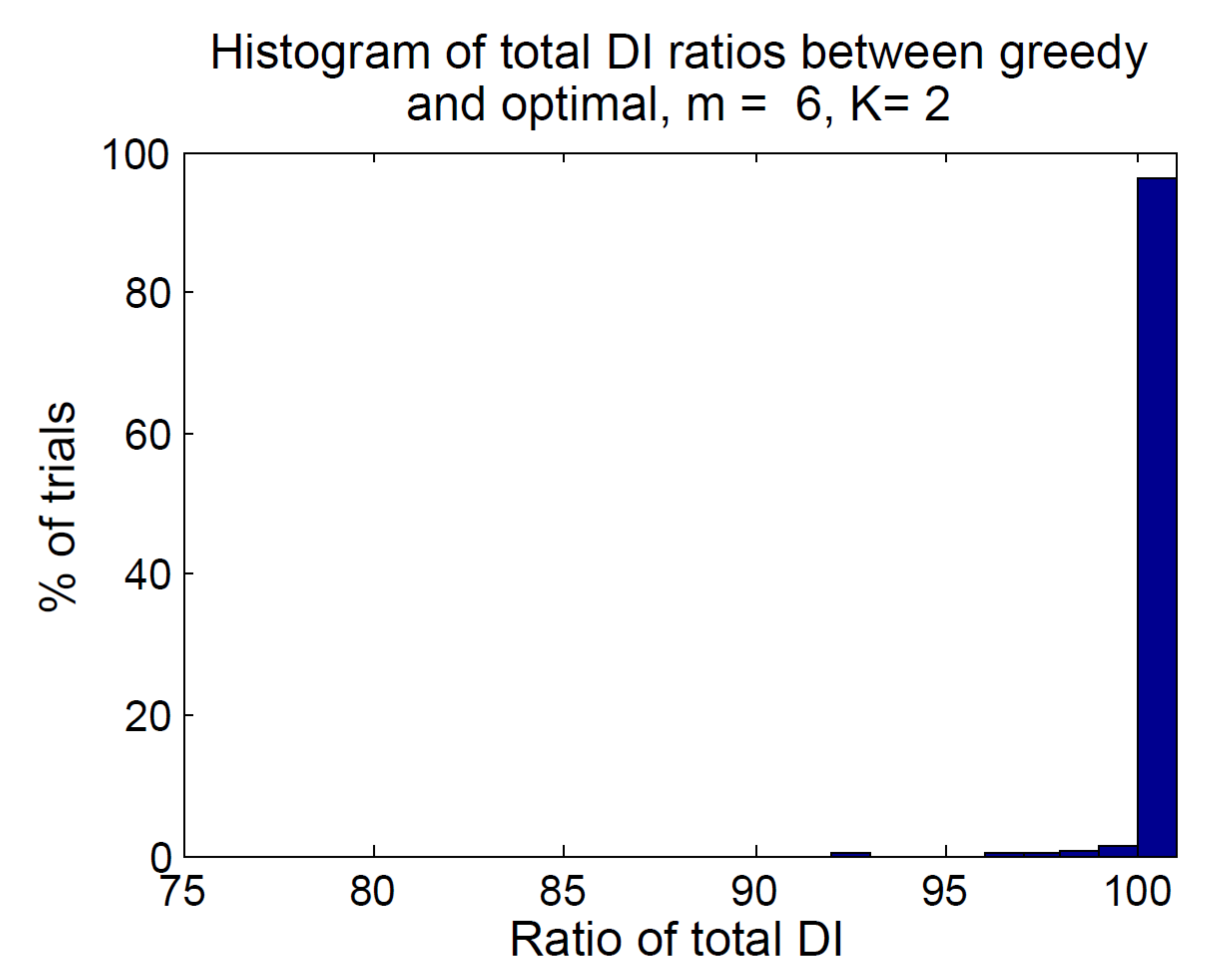}} \hspace{0.3cm}
  \subfigure[$m=15$.]{\label{fig:greedy_optim:m_15_K_4} \includegraphics[width=.40\columnwidth]{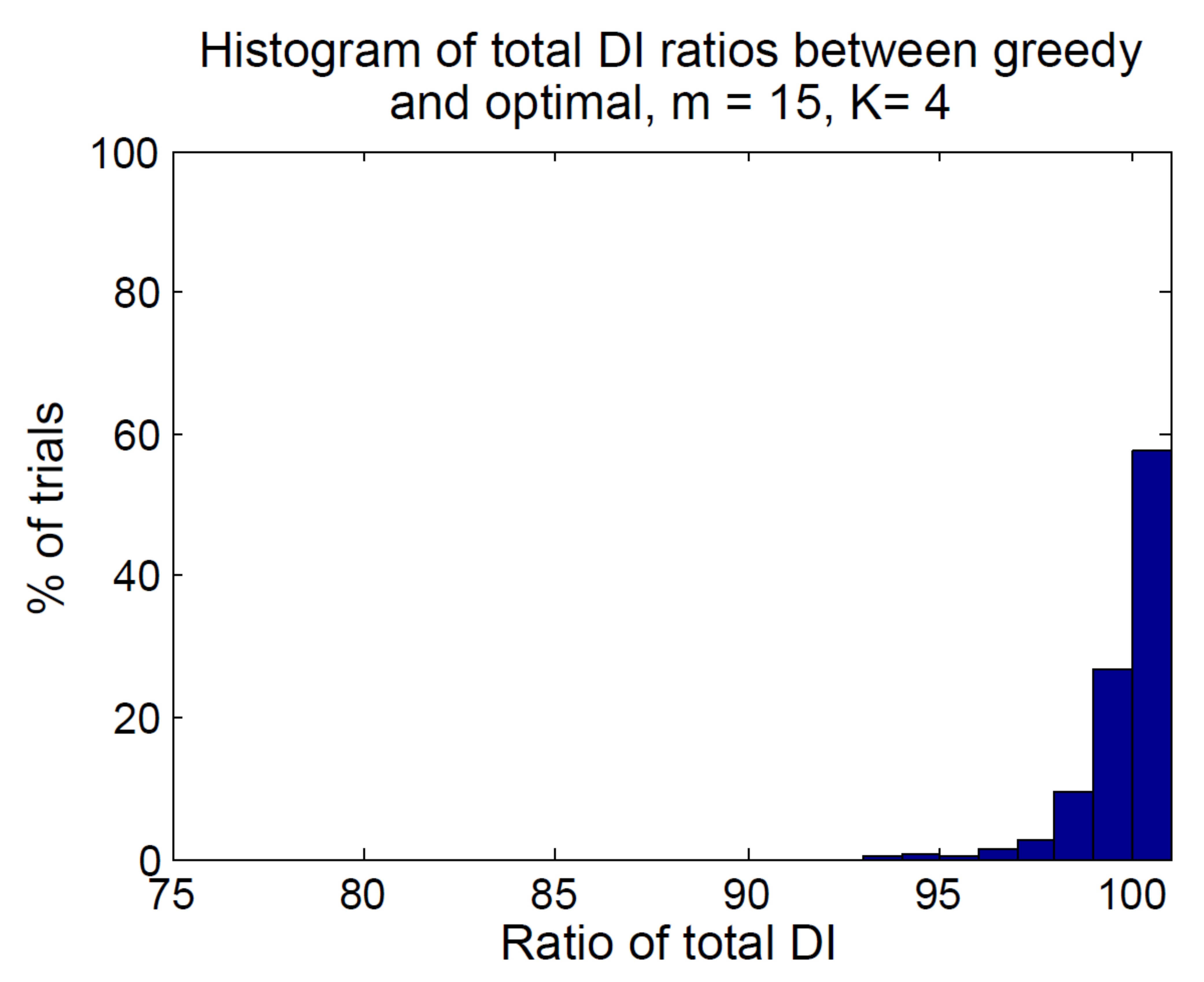}}  
  \caption{\small Histograms of the relative performance of Algorithms~1 and 3  using the ratio \eqref{eq:ratio_grd_opt}.  The right columns correspond to when both algorithms identified the same approximation.  
  }
  \label{fig:greedy_optim}
\end{figure}

\subsection{Comparison of Top-$r$ Approximations for Algorithms~1,2,3,4}

We next investigated how well the approximations from Algorithms~1, 2, 3, and 4 compared to the true parent sets.  We used modified versions of the algorithms to produce the top-$r$ approximations for each class, as discussed in Section~\ref{sec:best_r} (such as Algorithm~5).  For the MWDST algorithm we used \cite{choudhary2009edmonds}.

\subsubsection{Setup} 
The setup was similar to that described in Section~\ref{sec:sim:grdopt:setup}.  For the approximations, there were multiple in-degree $K$ values.  For $m=6$, in-degrees $K \in \{1, 2, 4\}$ were used, and for $m=15$, $K \in \{2, 4, 8\}$ were used.  Performance was measured by the ratio \beqa
\frac{\sum_{i=1}^m \I(\allX_{A(i)} \to \X_i) }{\sum_{i=1}^m  \I(\allX_{A_{\mathrm{True}}(i)} \to \X_i)}, \label{eq:sim:topR:ratio_appx_true}
\eeqa  where $A(i)$ denotes a parent set induced by an approximation, and $A_{\mathrm{True}}(i)$ is the true parent set. The ratio \eqref{eq:sim:topR:ratio_appx_true} was calculated exactly.  For each type of approximation, the top $r=10$ approximations were found and performance was averaged across trials.

\subsubsection{Results} 
Overall, the different approximations (unconstrained or connected, using optimal or greedy search) performed comparably for each of the $(m,K)$ pairs.  The results are shown in Figure~\ref{fig:avgDIrat_mKr}.  The most noticeable variation in performance was due to different in-degree $K$ values.  For both $m=6$ and $m=15$, increasing $K$, especially when $K$ was small, substantially improved performance.  

The connected approximations performed only slightly worse than the unconstrained.  For larger $m$ and $K$ the difference diminished.  For each $(m,K)$ pair, on average the approximations returned by the greedy search performed almost as well as that of the optimal.  This is consistent with the results in Section~\ref{sec:sim:grdopt:results} (see Figure~\ref{fig:greedy_optim}).  

Performance did not decay appreciably among the top approximations.  For $m=6$ and $K=1$, the performance difference between the first and tenth approximation is distinguishable, but for others it is not.

\begin{figure*}[t]
\centering
  \subfigure[$m = 6$, $K=1$.]{\label{fig:avgDIrat_mKr:m_6_K_1} \includegraphics[width=.30\columnwidth]{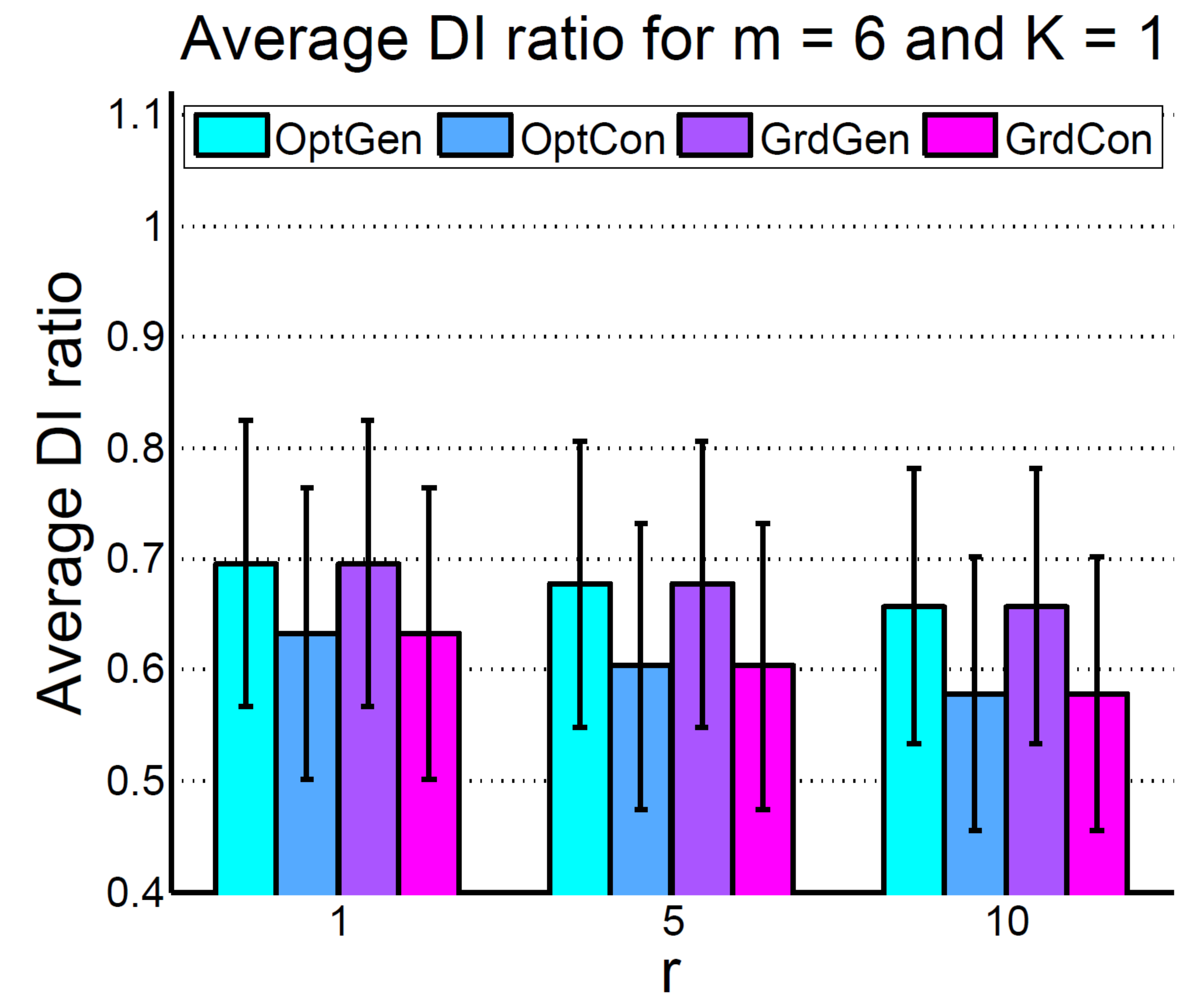}} 
  \hspace{0.1cm}
  \subfigure[$m=6$, $K=2$.]{\label{fig:avgDIrat_mKr:m_6_K_2} \includegraphics[width=.30\columnwidth]{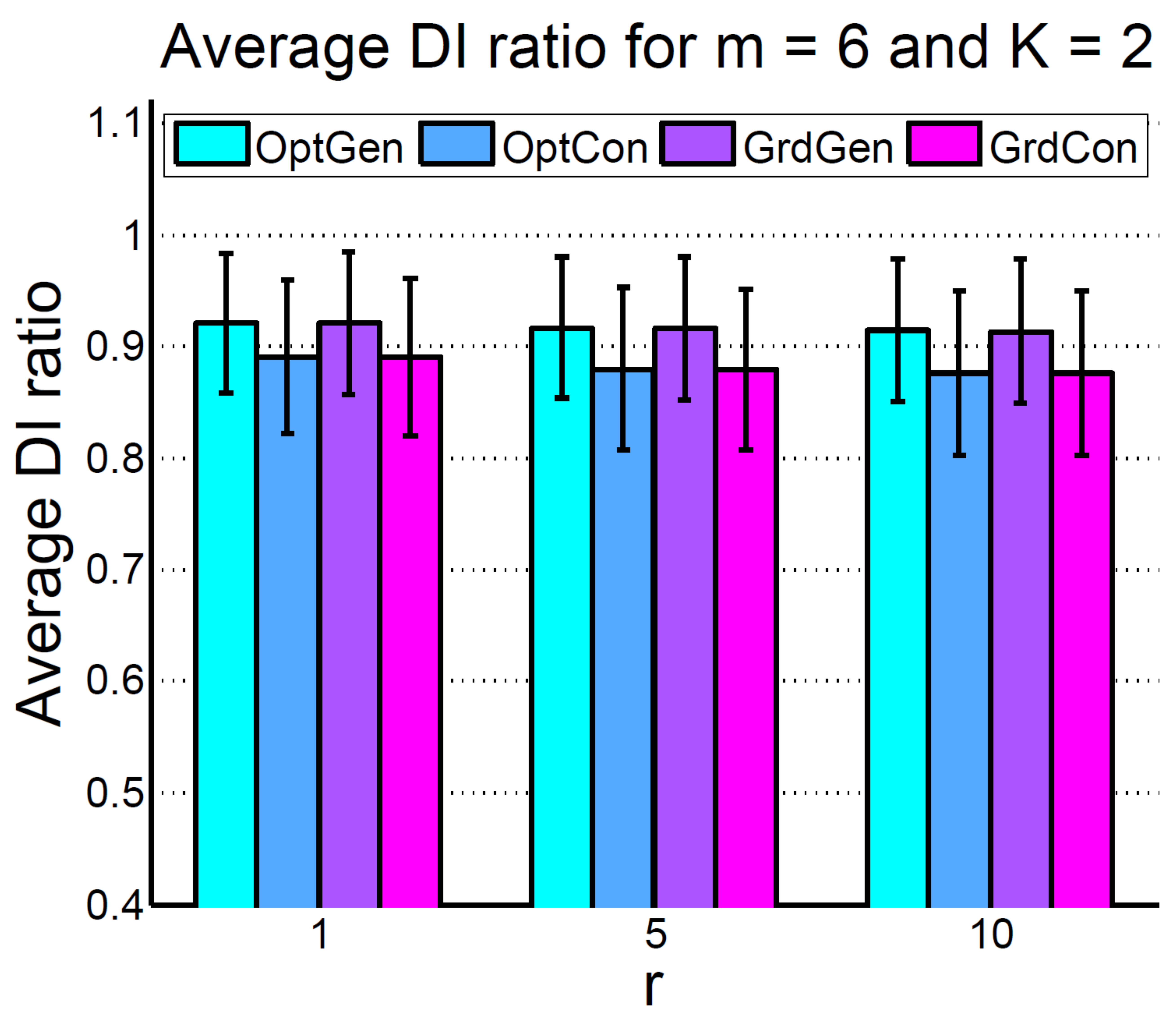}}  
    \hspace{0.1cm}
  \subfigure[$m=6$, $K=4$.]{\label{fig:avgDIrat_mKr:m_6_K_4} \includegraphics[width=.30\columnwidth]{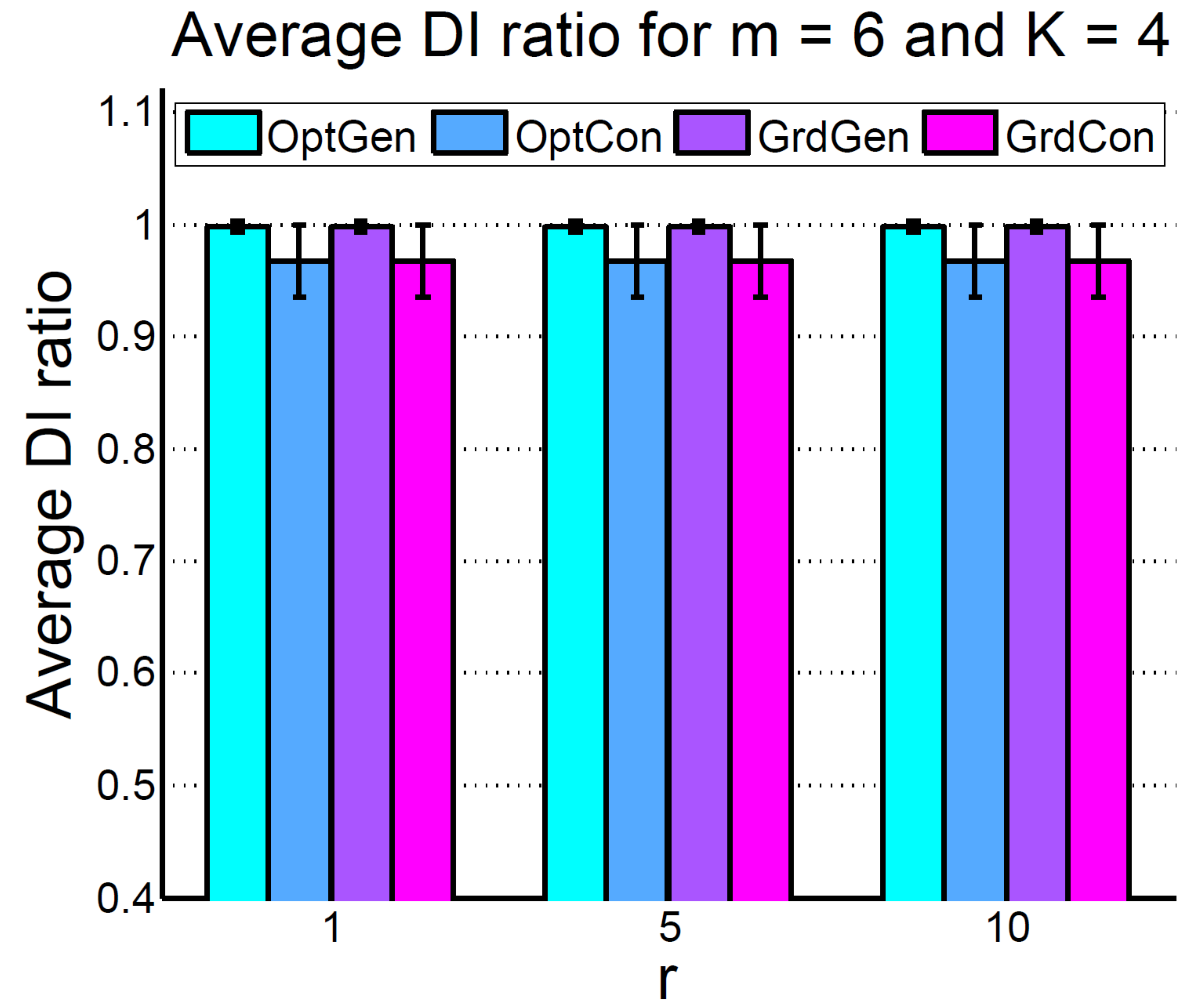}} 
    \\
  \subfigure[$m = 15$, $K=2$.]{\label{fig:avgDIrat_mKr:m_15_K_2} \includegraphics[width=.30\columnwidth]{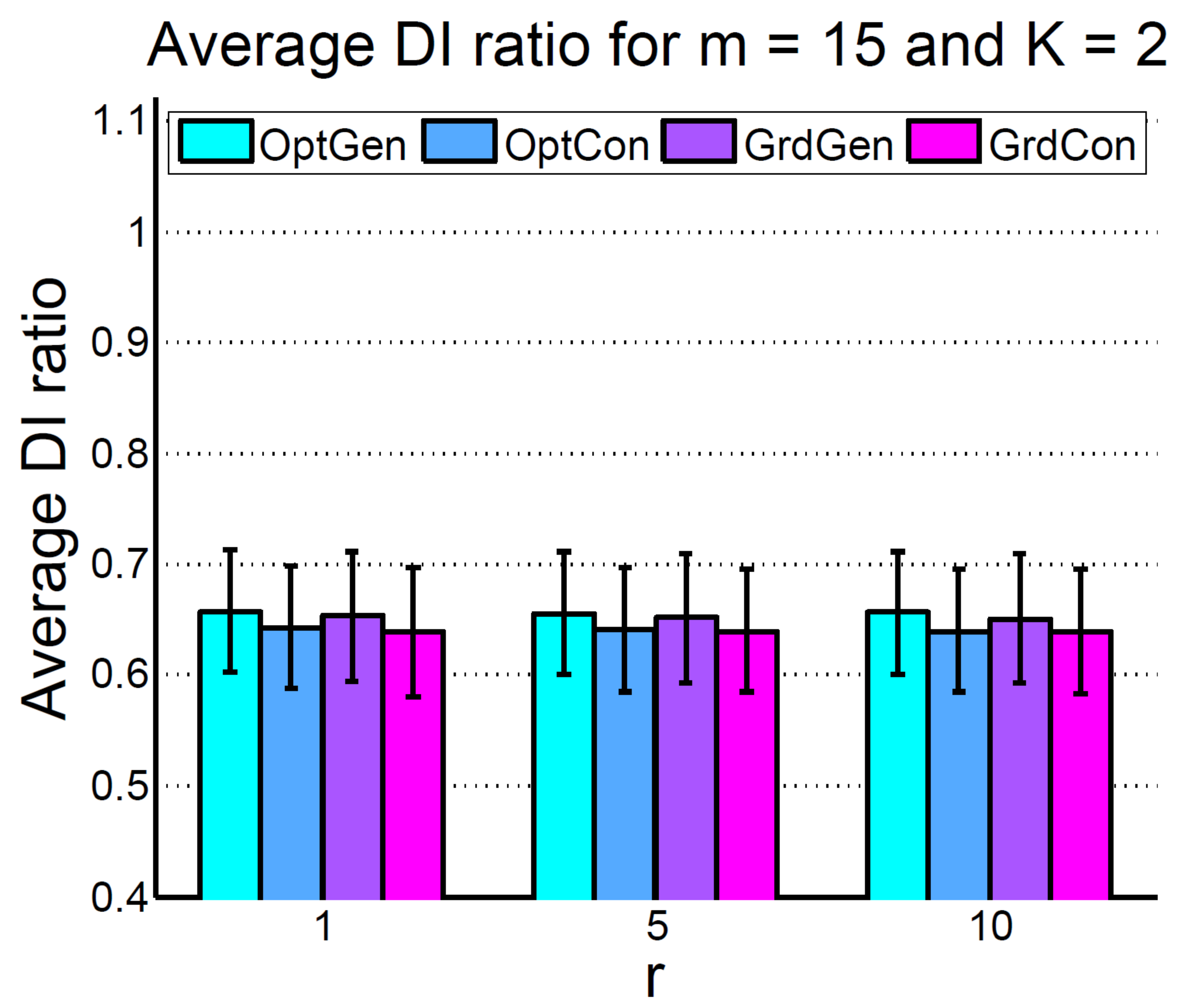}} 
  \hspace{0.1cm}
  \subfigure[$m=15$, $K=4$.]{\label{fig:avgDIrat_mKr:m_15_K_4} \includegraphics[width=.30\columnwidth]{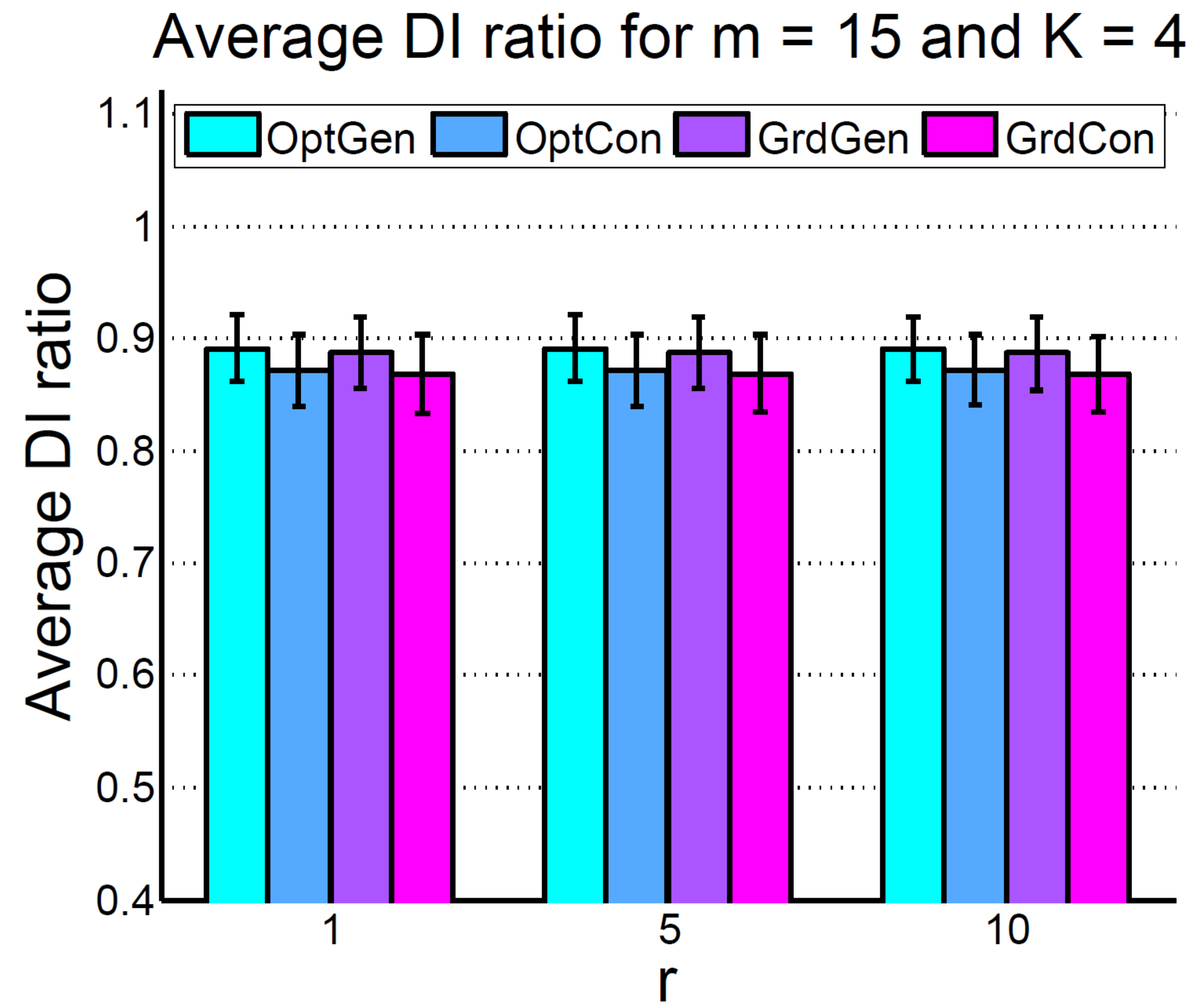}}  
    \hspace{0.1cm}
  \subfigure[$m=15$, $K=8$.]{\label{fig:avgDIrat_mKr:m_15_K_8} \includegraphics[width=.30\columnwidth]{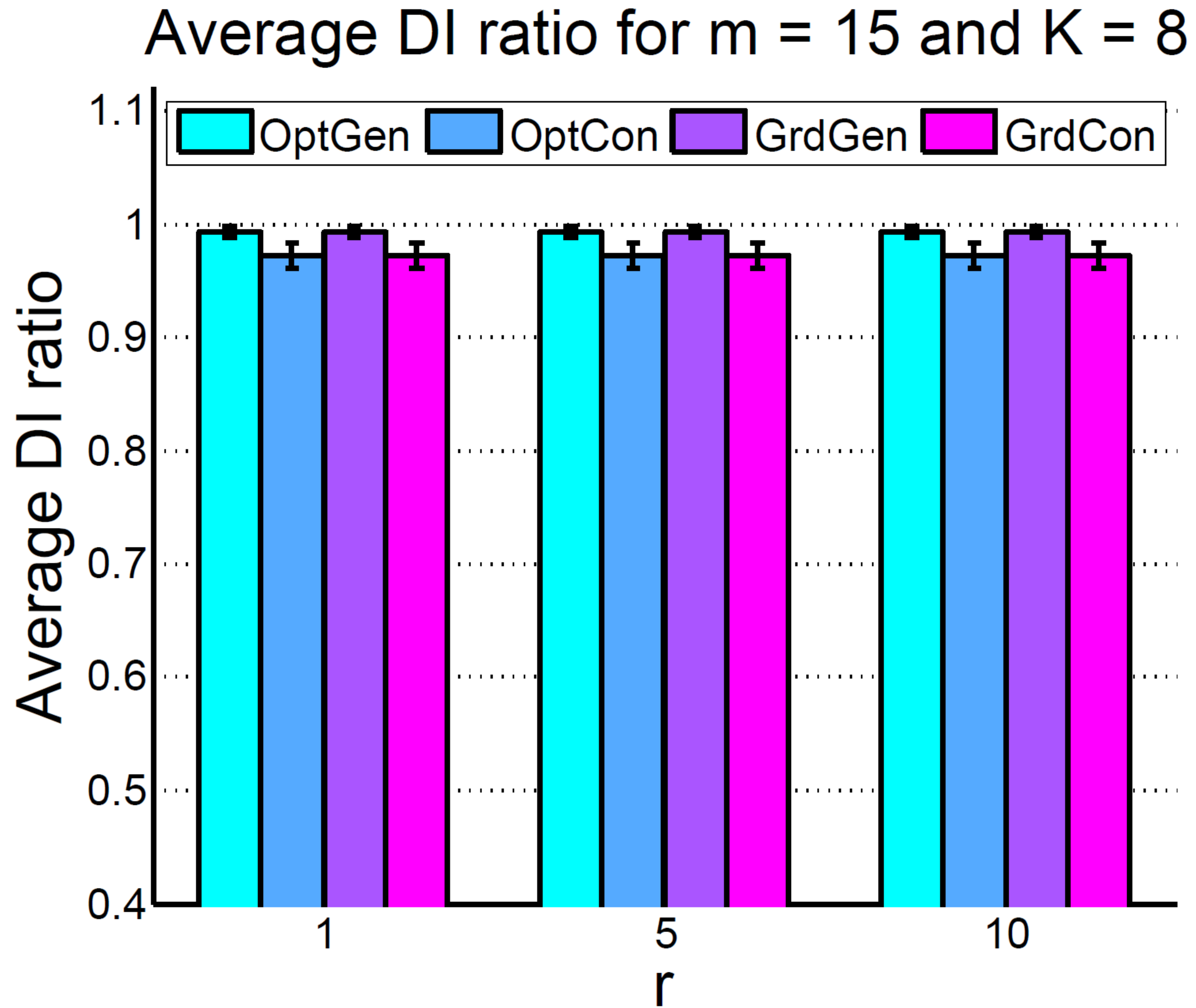}} 
  \caption{\small Plots of the average ratio \eqref{eq:sim:topR:ratio_appx_true} of approximations to the true parent sets.  Standard deviation error bars are shown.  For each type of search, optimal (``Opt'') and greedy (``Grd''), and each type of approximation, unconstrained (``Gen'') and connected (``Con''), the top-$r$ approximations are shown, with $r \in \{1, 5, 10\}$.  }
  \label{fig:avgDIrat_mKr}
\end{figure*}

Figure~\ref{fig:optgencon_r1234} shows diagrams of the top $r=4$ unconstrained and connected approximations for a single trial with $m=6$ and $K=2$.  Figures~\ref{fig:optgen_r1} and \ref{fig:optcon_r1} are the optimal unconstrained and connected approximations respectively.  The figures to the right show differences between the optimal and $r$th best approximations for $r\in \{2,3,4\}$.  Dashed gray edges are those removed and solid black ones are edges included.  For example, the third best unconstrained approximation, Figure~\ref{fig:optgen_r3}, has all of the same edges as the optimal, Figure~\ref{fig:optgen_r1}, except it has $\X_2 \to \X_4$ instead of $\X_1 \to \X_4$.

For both the optimal unconstrained and connected approximations, the several next best approximations had only minor changes.  Many edges were preserved.  For this particular trial, among the top four unconstrained approximations the only differences involved the parent of $\X_4$.  For the connected approximations, the second, third, and fourth best approximations mostly varied for the parent of $\X_5$.  The parents of $\X_1$, $\X_2$, and $\X_3$ were identical for all of the approximations shown.

\begin{figure*}[t]
\centering
  \subfigure[General, $r\!=\!1$.]{\label{fig:optgen_r1} \includegraphics[width=.20\columnwidth]{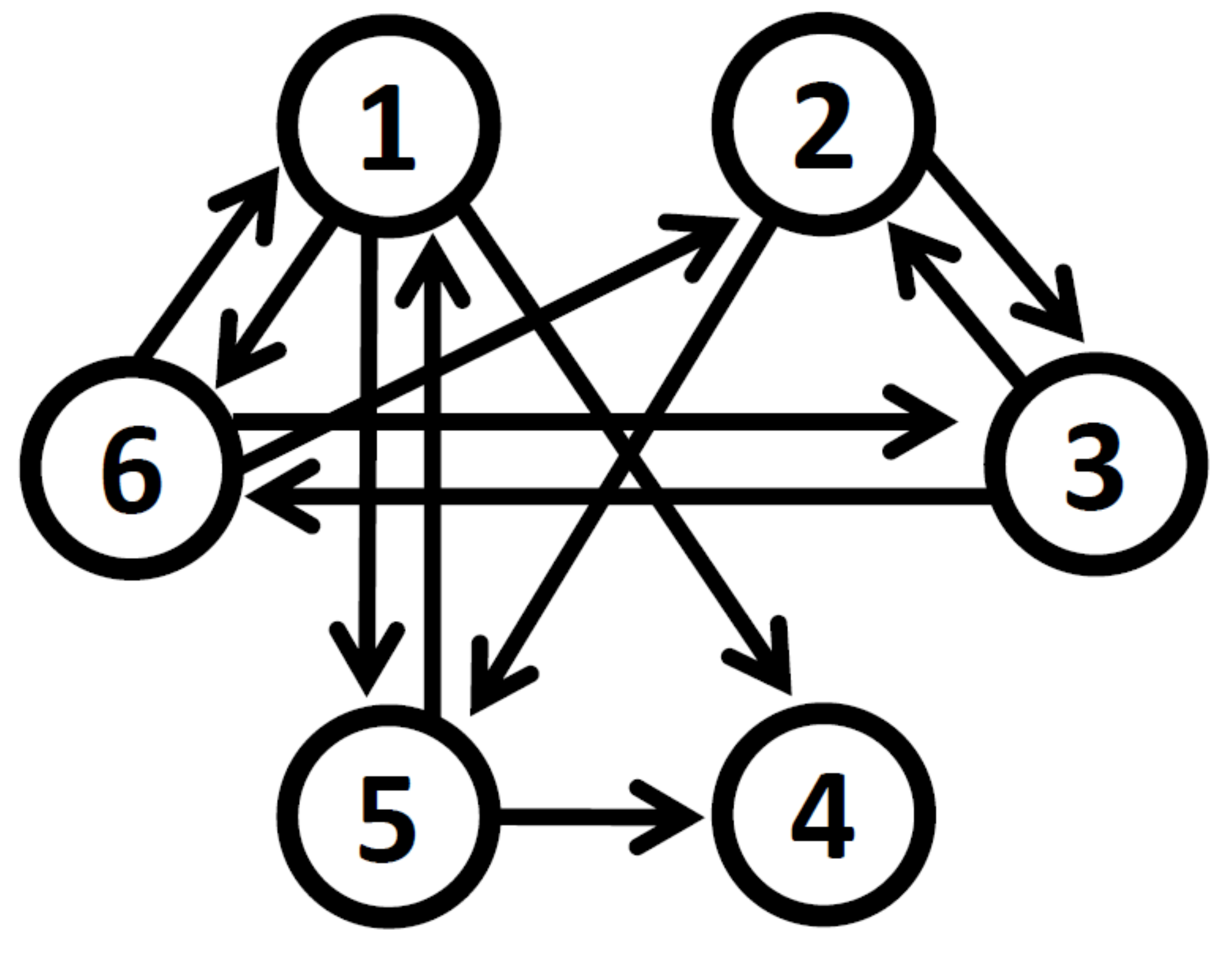}} 
  \hspace{0.1cm}
  \subfigure[General, $r\!=\!2$.]{\label{fig:optgen_r2} \includegraphics[width=.20\columnwidth]{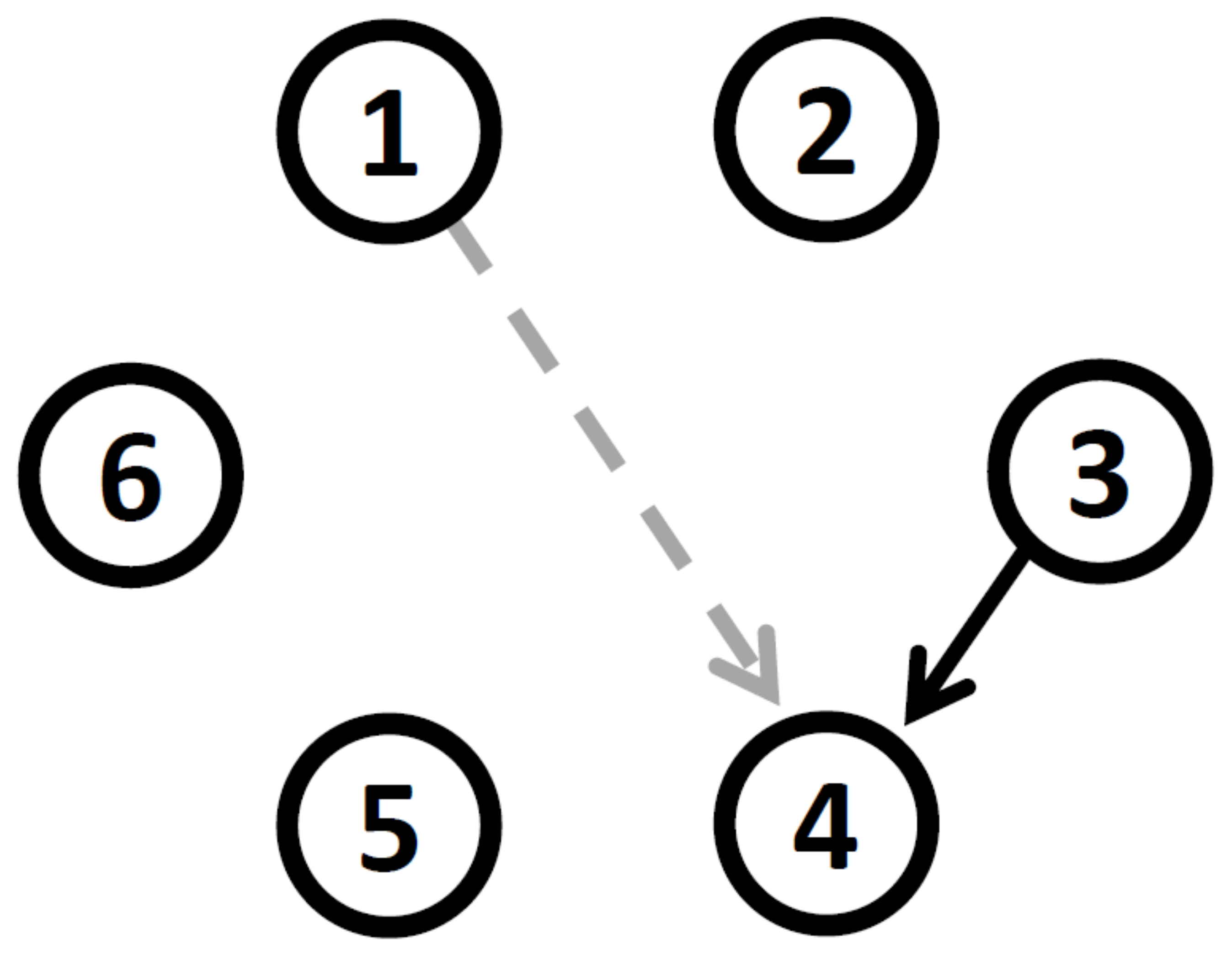}}  
    \hspace{0.1cm}
  \subfigure[General, $r\!=\!3$.]{\label{fig:optgen_r3} \includegraphics[width=.20\columnwidth]{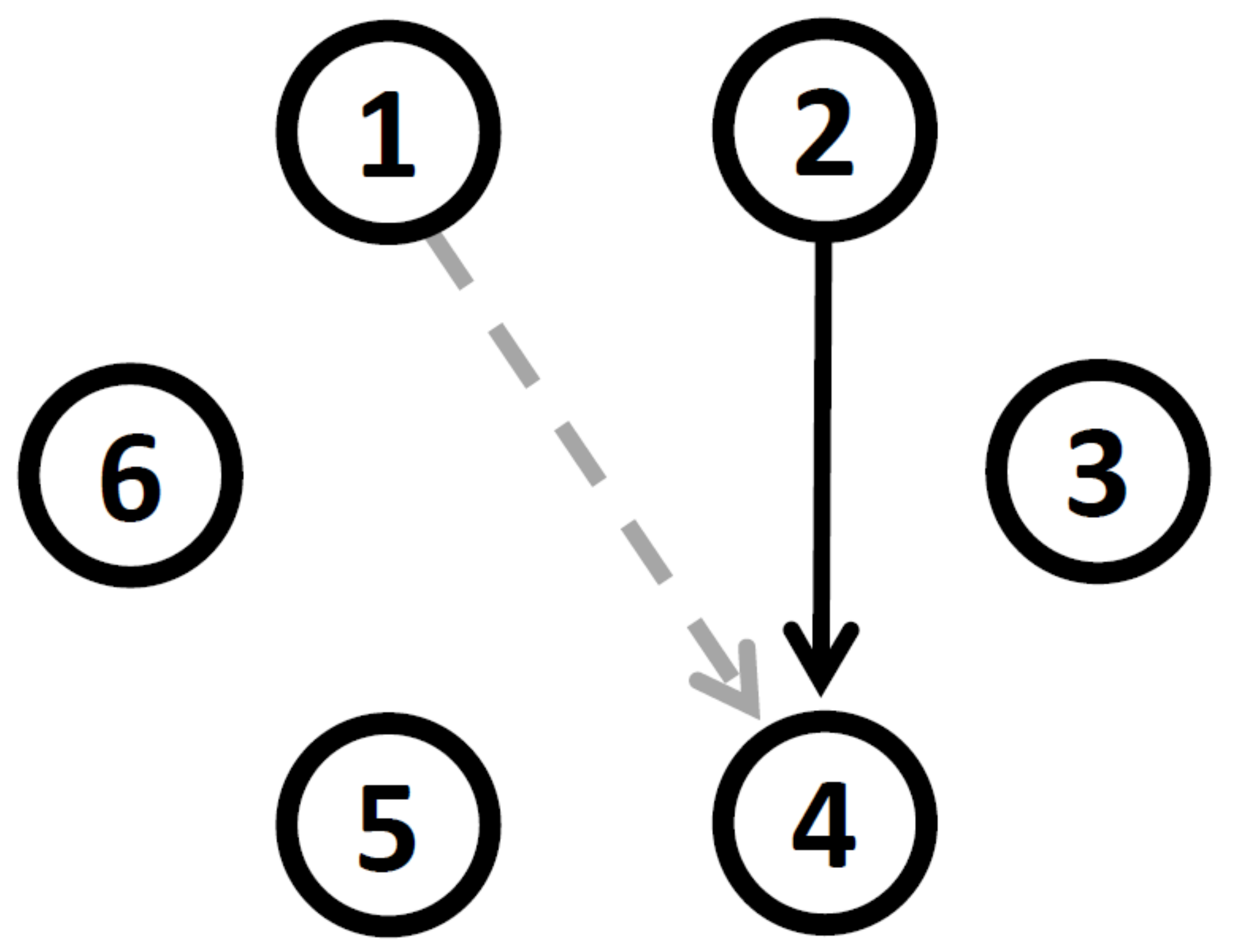}} 
      \hspace{0.1cm}
  \subfigure[General, $r\!=\!4$.]{\label{fig:optgen_r4} \includegraphics[width=.22\columnwidth]{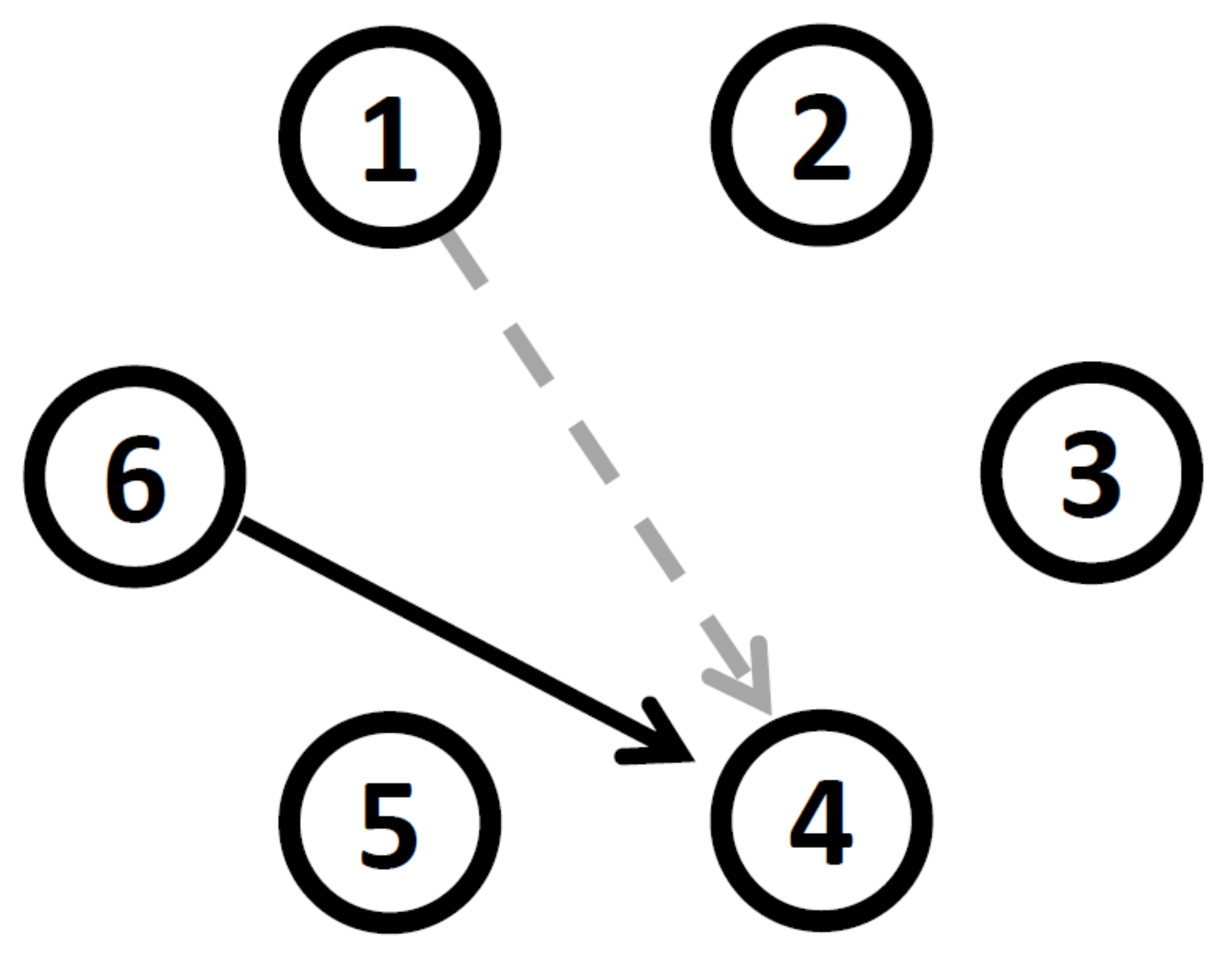}} 
    \\
  \subfigure[Connected, $r\!=\!1$.]{\label{fig:optcon_r1} \includegraphics[width=.22\columnwidth]{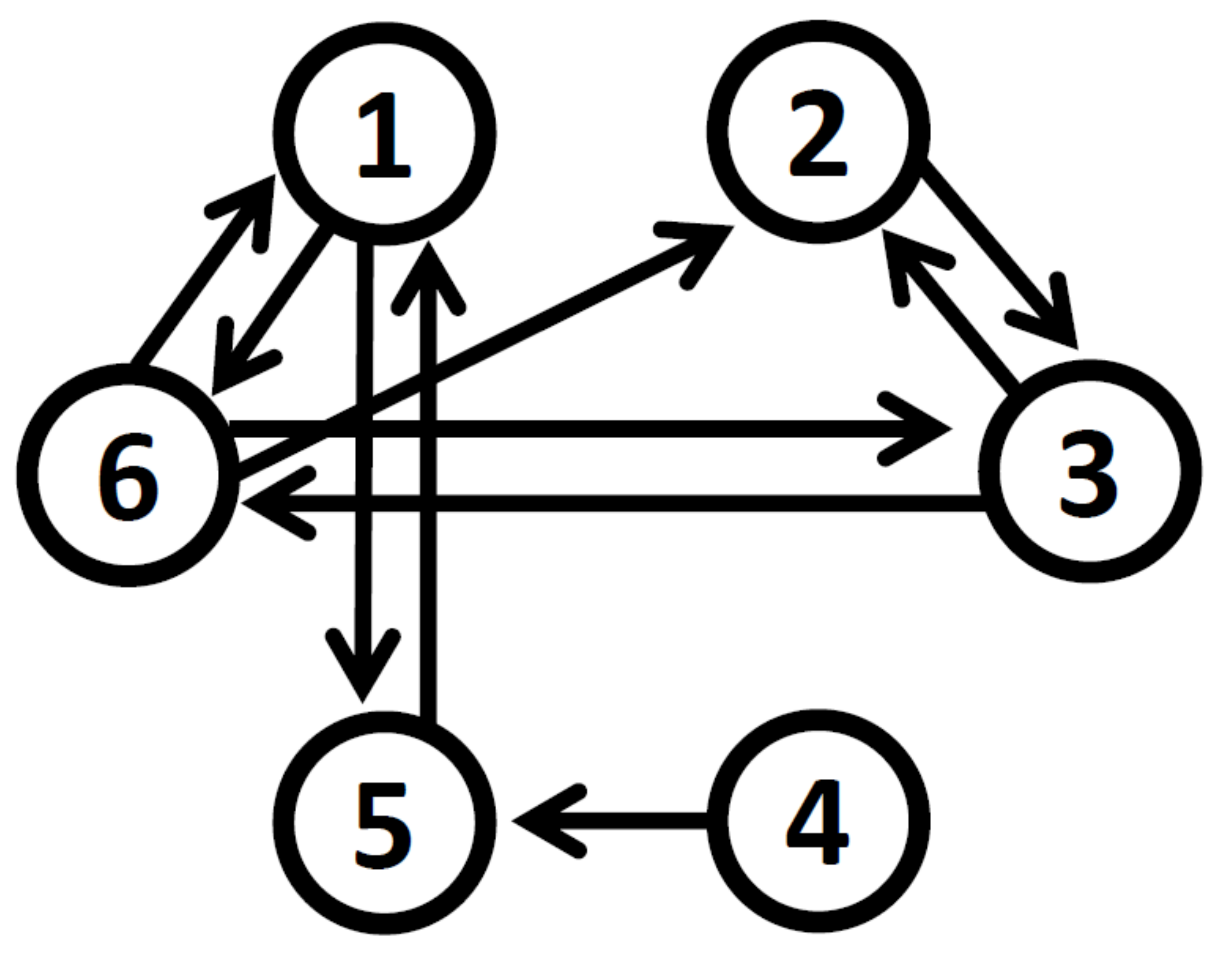}} 
  \hspace{0.1cm}
  \subfigure[Connected, $r\!=\!2$.]{\label{fig:optcon_r2} \includegraphics[width=.22\columnwidth]{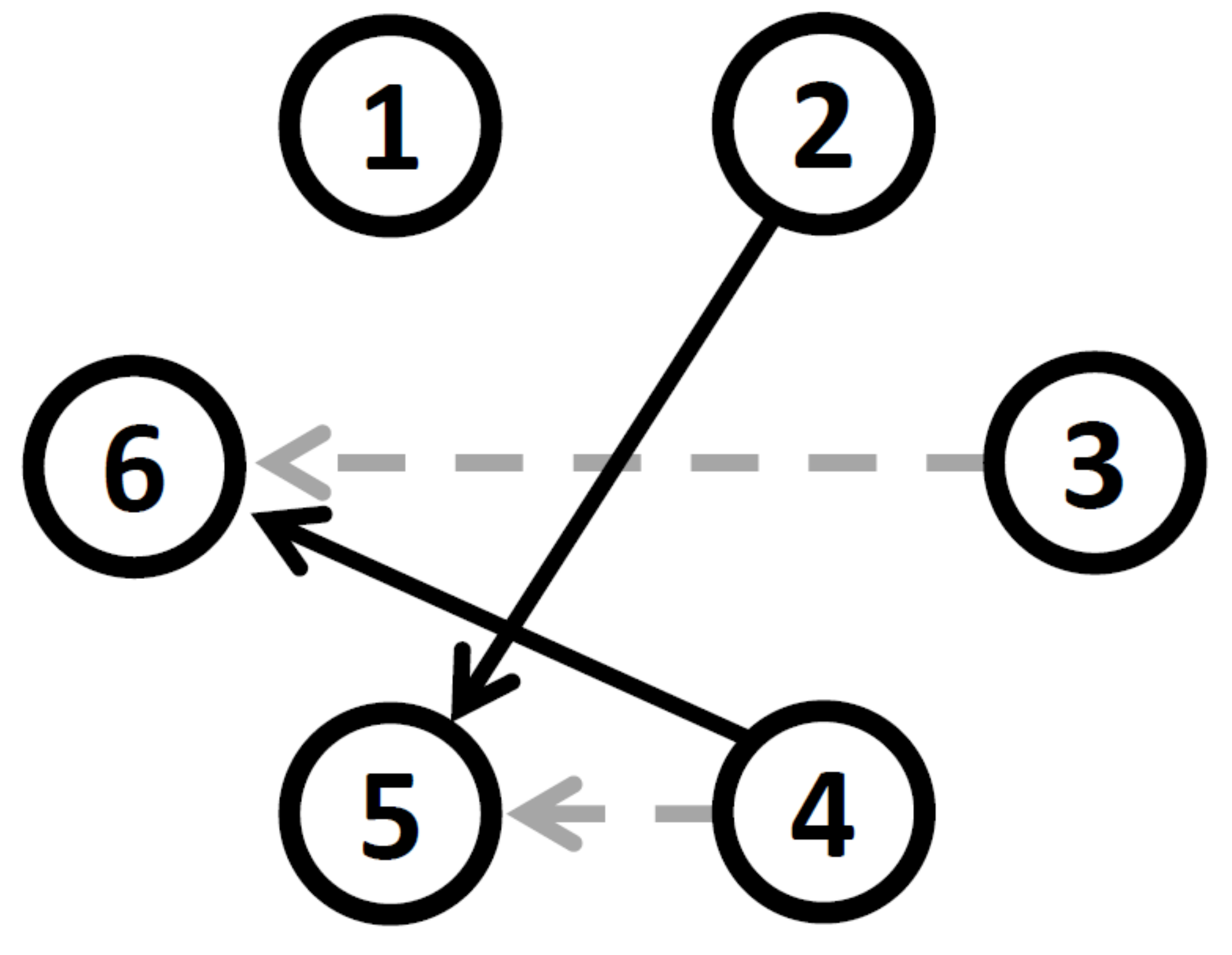}}  
    \hspace{0.1cm}
  \subfigure[Connected, $r\!=\!3$.]{\label{fig:optcon_r3} \includegraphics[width=.22\columnwidth]{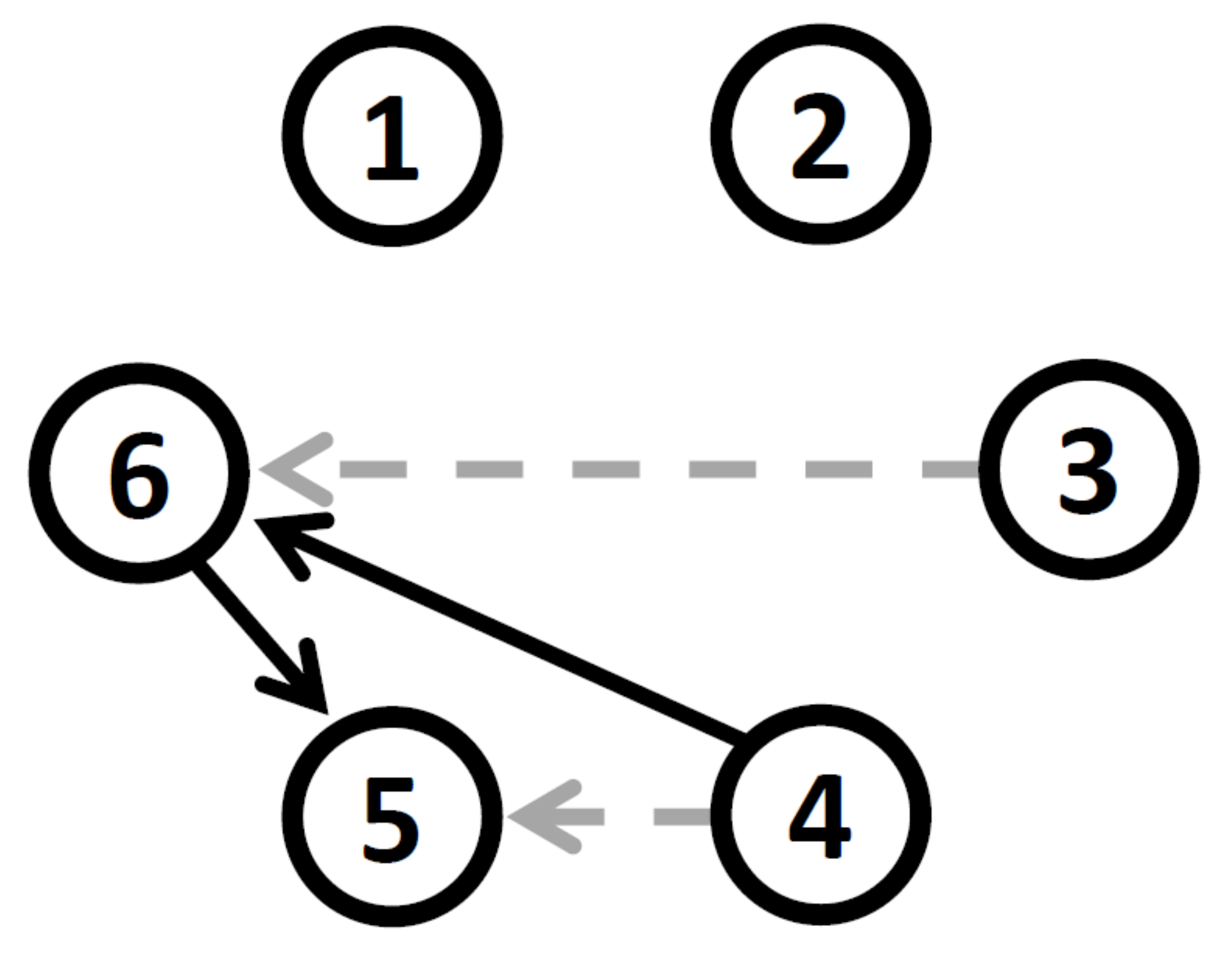}} 
      \hspace{0.1cm}
  \subfigure[Connected, $r\!=\!4$.]{\label{fig:optcon_r4} \includegraphics[width=.20\columnwidth]{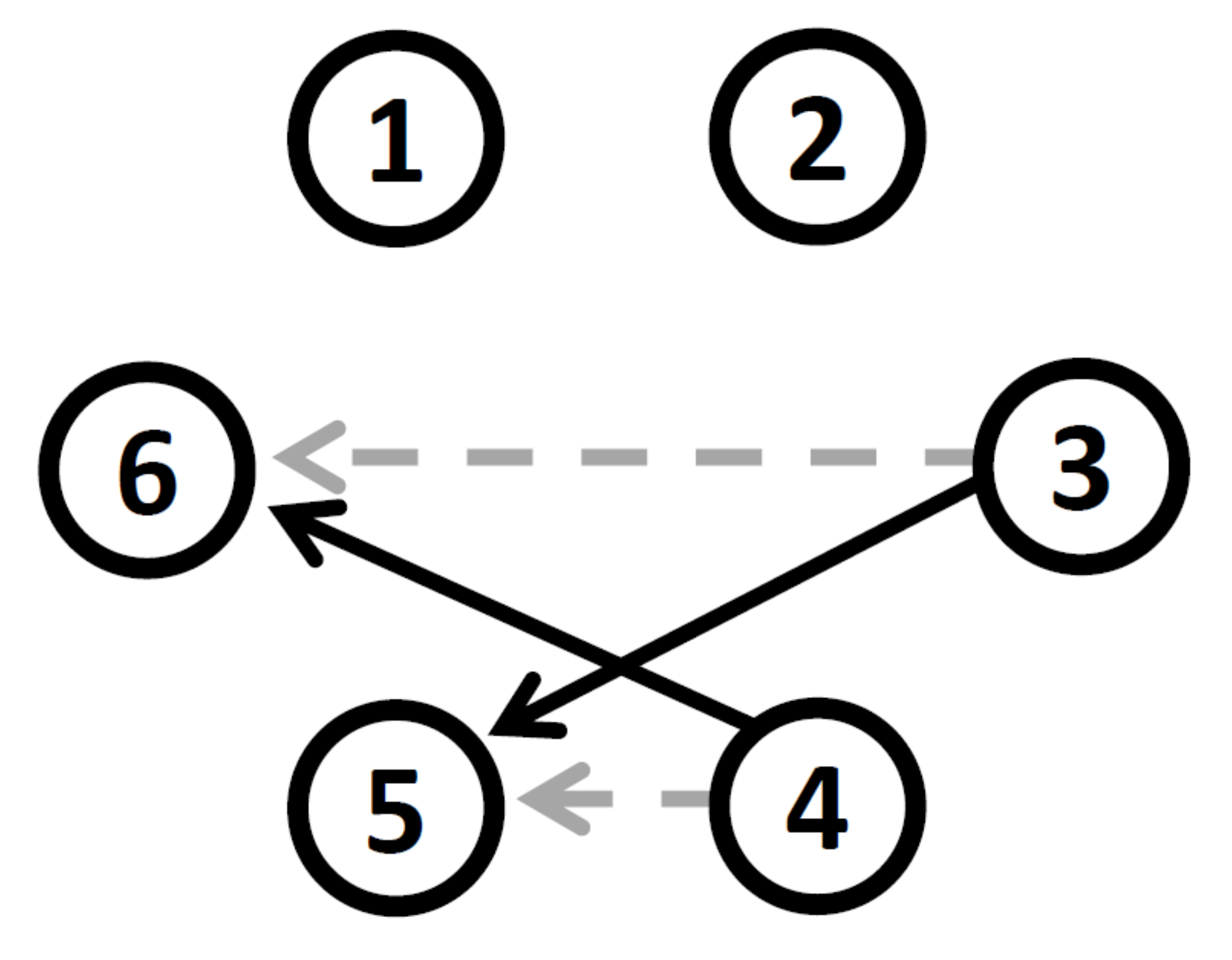}}  \caption{\small Diagrams of the top $r=4$  unconstrained and connected approximations using optimal search for a single trial with $m=6$ and $K=2$.  Figures~\ref{fig:optgen_r1} and \ref{fig:optcon_r1} are the optimal unconstrained and connected approximations respectively.  The figures to the right show differences between the optimal and $r$th best approximations.  Dashed edges were removed and solid edges were included.}
  \label{fig:optgencon_r1234}
\end{figure*}

\subsection{Performance Decay for Large $r$}

We also investigated the decay in performance quality as $r\to {m-1 \choose K}^m$.  

\subsubsection{Setup}

Using a simulation setup similar to Section~\ref{sec:sim:grdopt:setup}, 150 trials with $m=6$ processes and in-degree $K=2$ were run.  Unconstrained bounded in-degree approximations were obtained using Algorithm~5 and a similarly modified Algorithm~3 to obtain optimal and greedy search orderings respectively.  All approximations were computed, with $r = {m-1 \choose K}^m = 10^6$. The ratio \eqref{eq:sim:topR:ratio_appx_true} was used to measure performance.

\subsubsection{Results}

Figure~\ref{fig:3trialallr} depicts the ratio \eqref{eq:sim:topR:ratio_appx_true} for all approximations, ordered using Algorithm~5, for three trials.  Due to estimation error, the actual ratio values did not decay monotonically.  Figure~\ref{fig:3trialallr_est} plots the estimated values.

\begin{figure}[t]
\centering
  \subfigure[Actual ratio values.]{\label{fig:3trialallr} \includegraphics[width=.45\columnwidth]{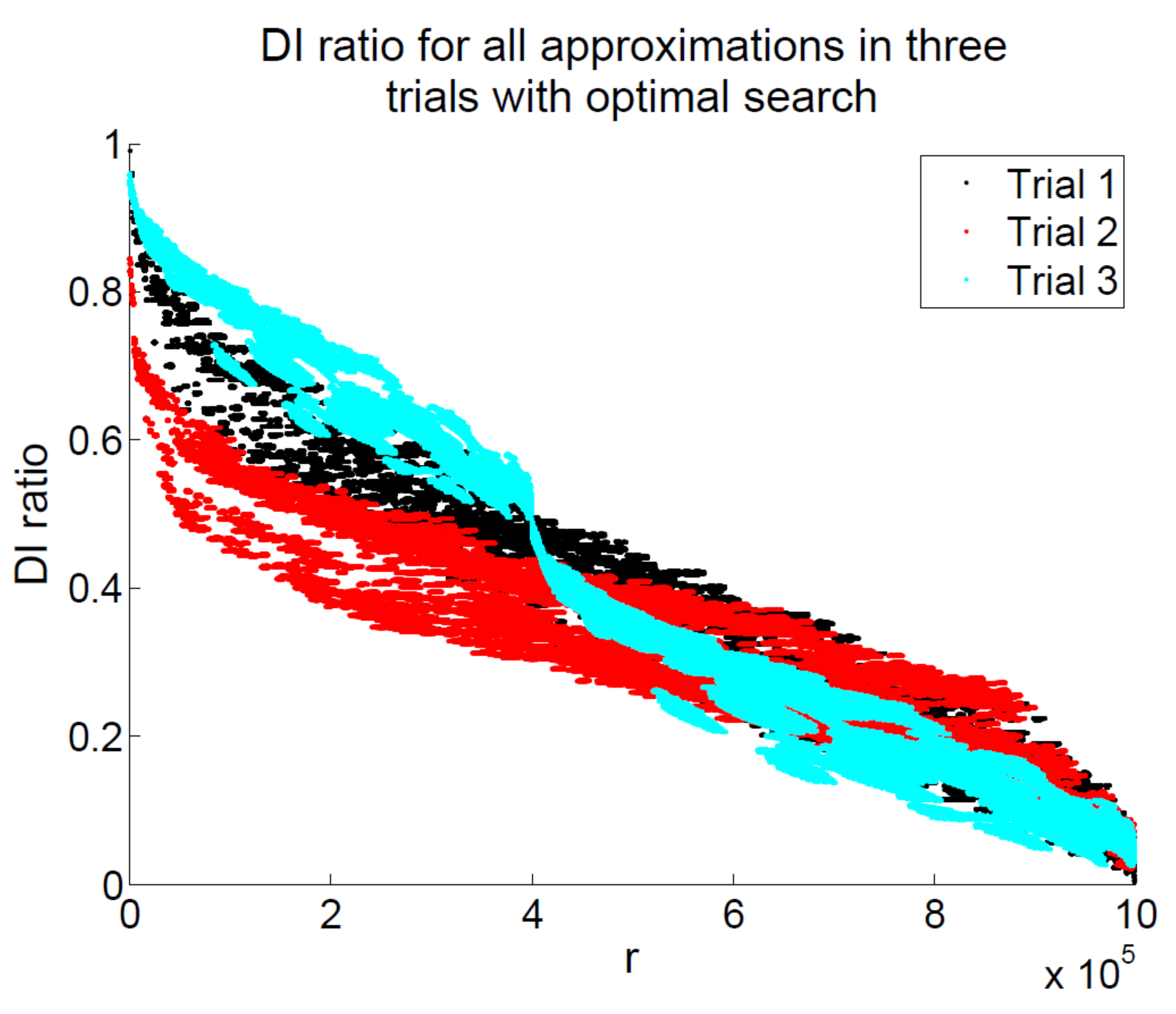}} 
  \quad
  \subfigure[Estimates used by Algorithm~5, scaled.]{\label{fig:3trialallr_est} \includegraphics[width=.45\columnwidth]{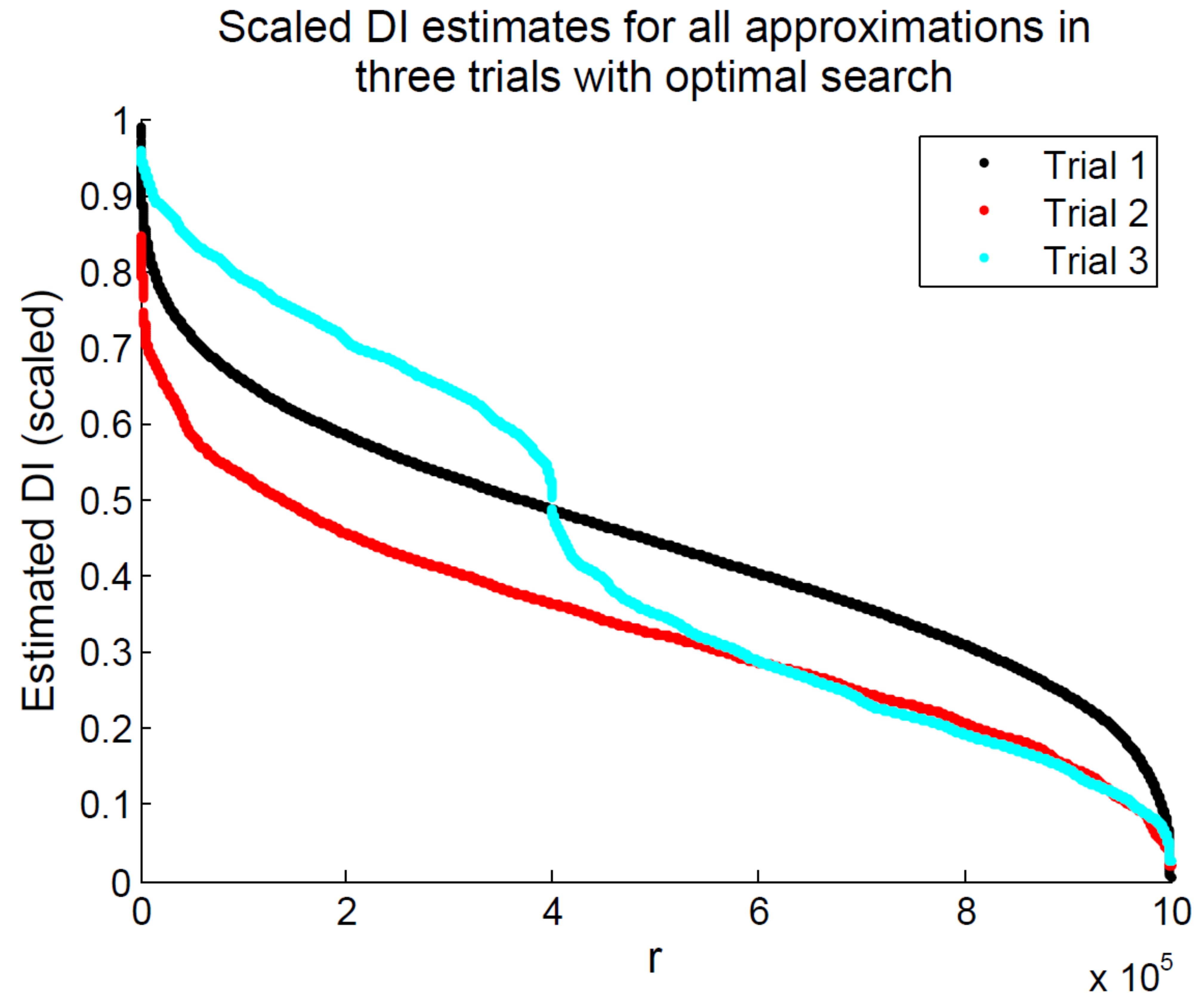}} 
   \caption{\small The values of the ratio \eqref{eq:sim:topR:ratio_appx_true} for all approximations in the order selected by Algorithm~5.  Three trials are shown.  Figure~\ref{fig:3trialallr} shows the actual ratio values.  The spread is due to estimation error. Figure~\ref{fig:3trialallr_est} shows the estimated values (normalized to the scale of Figure~\ref{fig:3trialallr}). }
  \label{fig:allappx3trial}
\end{figure}

\begin{figure}[t]
	\centering
	\includegraphics[width=.45\columnwidth]{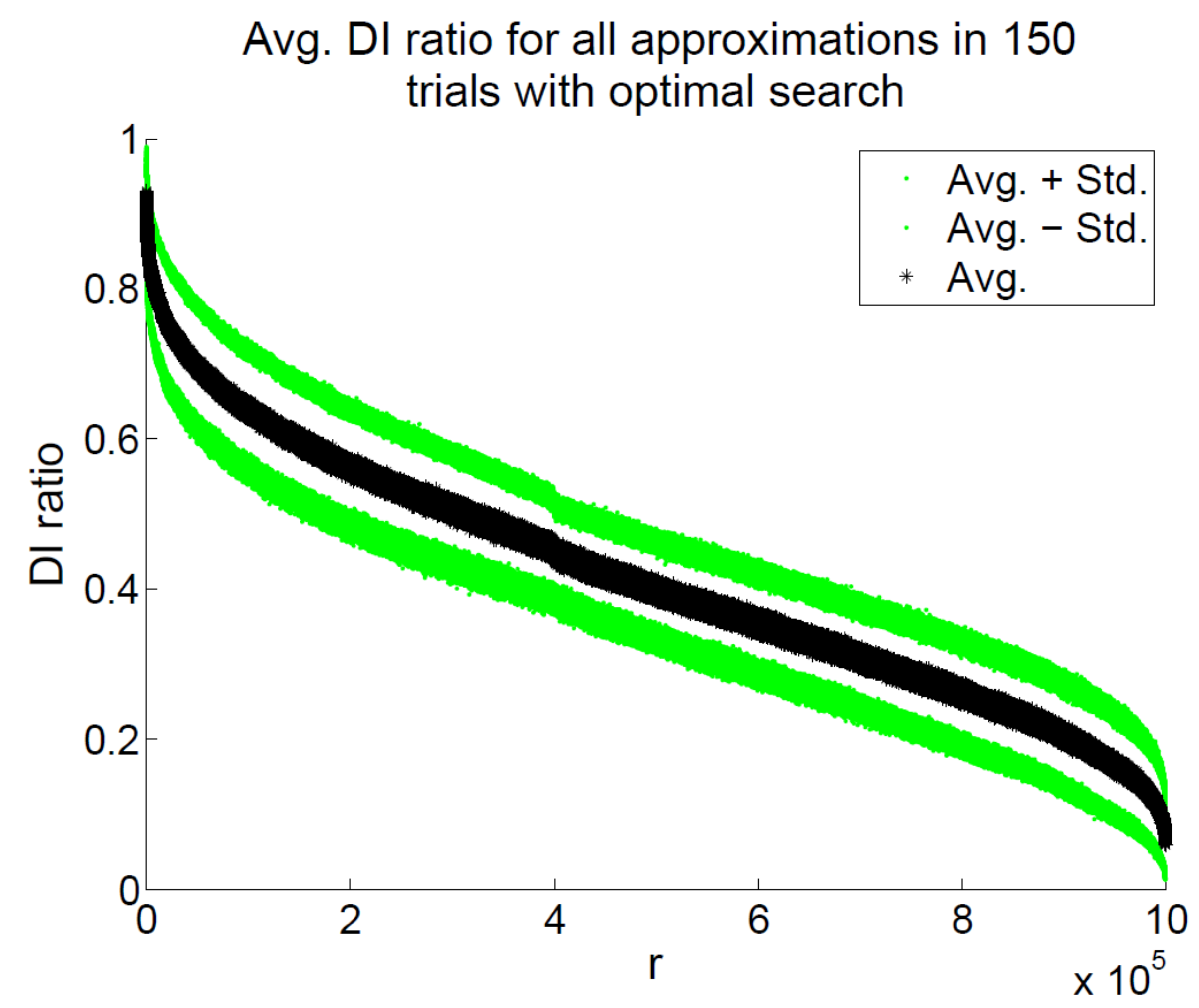}
   \caption{\small The ratio \eqref{eq:sim:topR:ratio_appx_true} for all approximations in the order selected by Algorithm~5, averaged over 150 trials.  The black and green curves depict the mean standard deviation respectively.}
   \label{fig:avgDI_allapx} 
\end{figure}

Figure~\ref{fig:avgDI_allapx} shows the true ratio \eqref{eq:sim:topR:ratio_appx_true} value using the ordering returned by Algorithm~5, averaged over 150 trials.  The mean and one standard deviation are represented by the black and green curves respectively.  Consistent with the curves for individual trials in Figure~\ref{fig:allappx3trial}, the shape of the curves in Figure~\ref{fig:avgDI_allapx} are similar to a logit function.  There was a sharp decrease in the quality of approximation in the low and high $r$ regimes, with a nearly linear decay for most $r$.

\begin{figure}[t]
	\centering
  \includegraphics[width=.45\columnwidth]{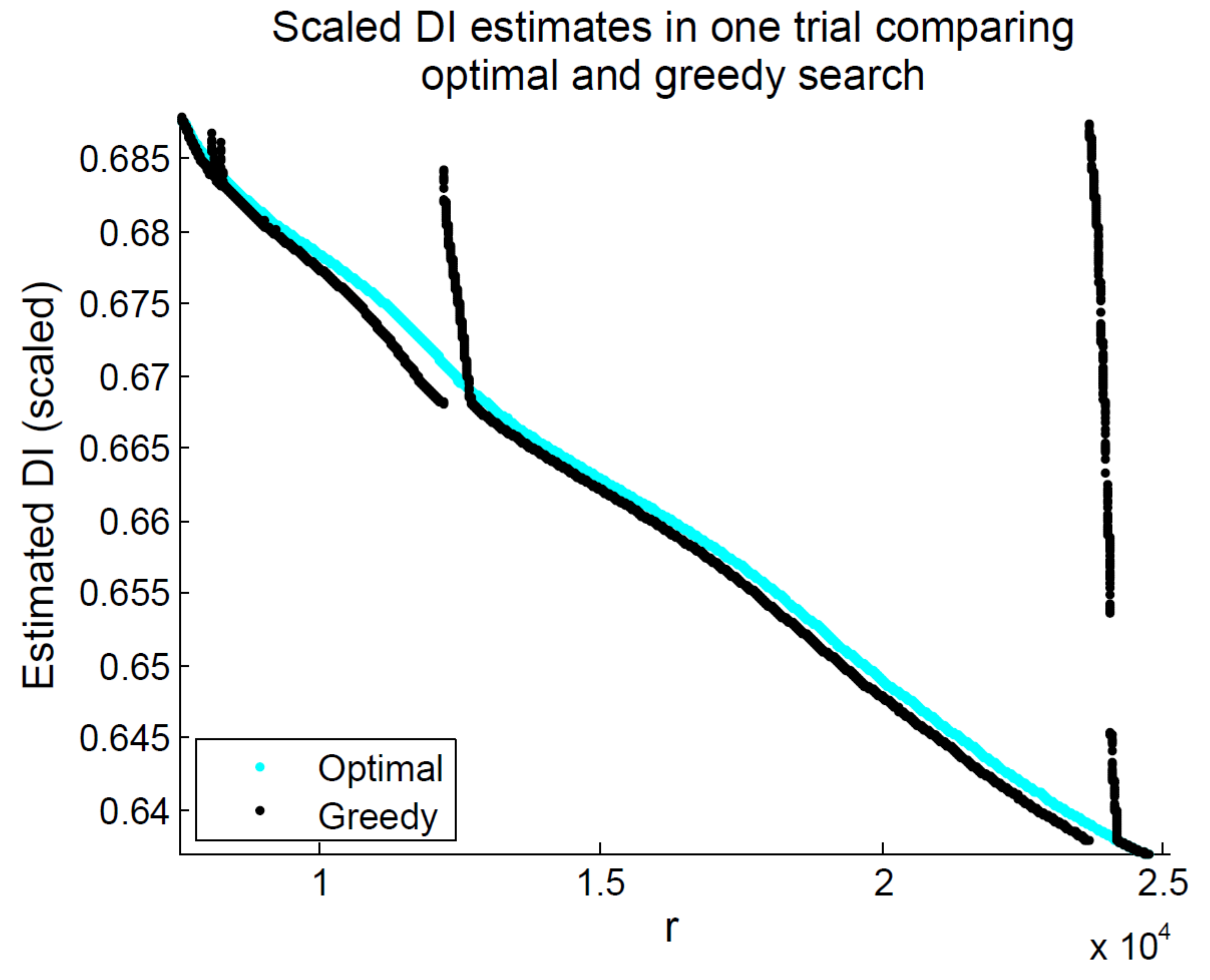}
  \caption{\small The estimate values used by Algorithm~5 and a modified Algorithm~3 to rank unconstrained approximations using optimal and greedy search.  One trial is shown. Estimates are scaled to actual ratio values \eqref{eq:sim:topR:ratio_appx_true}.  The smooth light blue curve corresponds to ordering from the optimal search.  The black curve corresponds to the greedy search.}
  \label{fig:grd_jumps} 
\end{figure}

The greedy search, using a modified Algorithm~3, performed comparably to the optimal search, Algorithm~5.  For $r=1$, this was shown in Figure~\ref{fig:greedy_optim}, and for $r\leq10$ this was shown in Figure~\ref{fig:avgDIrat_mKr}.  However, the current analysis confirmed that even for large $r$ this was true.  The analogous Figure~\ref{fig:avgDI_allapx} for the greedy ordering was visually indistinguishable and so not shown.  However, effects of greedy ordering were clearly seen in some trials.  Figure~\ref{fig:grd_jumps} depicts estimate values for approximations in one trial, with the light blue monotonic curve for the optimal ordering and the discontinuous black curve depicting the greedy ordering.  The large jumps are characteristic of the depth-first search, as the worst-case along one branch is worse than the best-case of the next, although the former branch was initially more promising. However, such large discontinuities as shown in Figure~\ref{fig:grd_jumps} were rare overall.  The ordering returned by the greedy search was largely similar to the optimal ordering.

\begin{figure}[t]
 \centering 
 \includegraphics[width=.45\columnwidth]{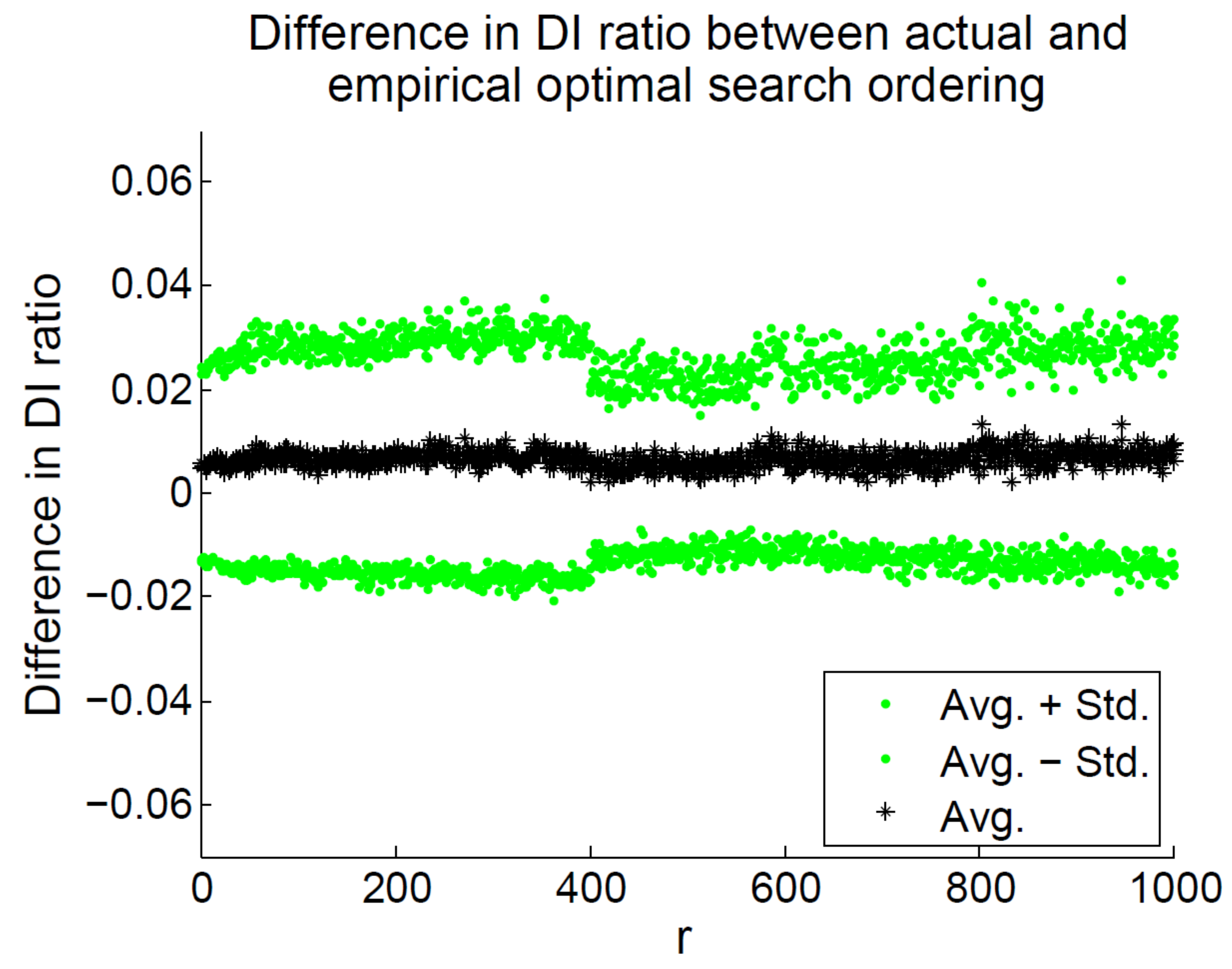} 
 \caption{\small The difference between the ratio \eqref{eq:sim:topR:ratio_appx_true} for the actual optimal ordering of unconstrained approximations and the ordering returned by Algorithm~3 (using estimates).  Results are averaged over 150 trials.  The black curve depicts the mean.  The green curves depict one standard deviation.   }
 \label{fig:diff_ordering} 
\end{figure}

As illustrated in Figure~\ref{fig:allappx3trial}, estimation errors led to errors in the optimal search ordering returned by Algorithm~5.  It is useful to characterize how large those ordering errors were.  Figure~\ref{fig:diff_ordering} shows the error, measured by the difference of the ratio \eqref{eq:sim:topR:ratio_appx_true} between the actual $r$th best approximation and the $r$th approximation in the ordering returned by Algorithm~5 for $r\leq 1000$. The results were averaged over 150 trials.  In terms of percentage points of the sum of directed information of the true parent sets to children, the denominator of \eqref{eq:sim:topR:ratio_appx_true}, Algorithm~5 performed well despite the estimation error.  The mean difference was slightly biased above $0$ and the standard deviations were within $4\%$ for all $r\leq 1000$.  The slight bias in the mean was expected; for small $r$, the approximations were mostly replaced with worse approximations, if any.  Overall, the greedy approximations performed nearly as well as the optimal approximations.

\section{Conclusion} \label{sec:concl}

In this paper, we presented several novel methods related to approximating directed information graphs.  The approximations allowed substantial flexibility for users.  Each algorithm took the in-degrees of the nodes as input.  With larger in-degrees, the approximations became better, but at the cost of visual simplicity of the graph and computational efficiency.  The approximation could be unconstrained or connected, and found using an optimal or a more efficient, near-optimal greedy search.  Furthermore, one could generate the top $r$ solutions, not only the best.  This enabled evaluation for which edges are most significant as well as finding the best solution of a more constrained class of topologies.  Lastly, the empirical results demonstrated the utility of these methods, especially showing that on average the greedy search performed much better than the worst-case lower bound.


\acks{C.\ J.\ Quinn was supported by the Department of Energy Computational Science Graduate Fellowship, which is provided under Grant DE-FG02-97ER25308.  He completed this work at the Department of Electrical and Computer Engineering, Coordinated Science Laboratory, University of Illinois, Urbana, Illinois 61801.  A.\ Pinar was supported by the DOE ASCR Complex  Distributed Interconnected Systems  (CDIS) program, the GRAPHS Program at DARPA, and the Applied Mathematics Program at the U.S. Department of Energy.  Sandia National Laboratories is a multi-program laboratory managed and operated by Sandia Corporation, a wholly owned subsidiary of Lockheed Martin Corporation, for the U.S. Department of Energy's National Nuclear Security Administration under contract DE-AC04-94AL85000. }




\appendix

\section{Proof of Theorem~\ref{thm:alg2:opt_con}} \label{apdx:thm:alg2:opt_con}

\begin{proof} Let $\calT$ be the set of all directed spanning trees on $m$ nodes.  For a given tree $T\in \calT$, let $\widetilde{\calG}_K^T \subseteq \widetilde{\calG}_K $ denote the set of approximations $\PT \in \widetilde{\calG}_K $ that contain $T$ as a directed spanning tree subgraph.  Every $\PT \in \widetilde{\calG}_K $ contains at least one such $T\in \calT$ as a subgraph, so $\widetilde{\calG}_K = \bigcup_{T \in \calT} \widetilde{\calG}_K^T$.  

For any tree $T \in \calT$, the best approximation $\PT^T \in \widetilde{\calG}^T_K$ is the one that for every edge $\X_{a(i)} \to \X_i$ in $T$, sets $\tildeA(i,a(i))$ as the parent set for node $i$.  This follows from \eqref{eq:algs:optconAtil}, since $a(i) \in \tildeA(i,a(i))$, so $T$ will be a subgraph of this approximation, and the sets $\tildeA(i,a(i))$ are the best such parent sets.  Thus, 
\beqa \max_{\PT \in \widetilde{\calG}_K} \sum_{i=1}^m \I(\allX_{A(i)} \to \X_i) 
&=& \max_{T \in \calT} \max_{\PT \in \widetilde{\calG}^T_K} \sum_{i=1}^m \I(\allX_{A(i)} \to \X_i) \label{eq:prf:Alg2:1} \\
&=& \max_{T \in \calT} \sum_{i=1}^m \I(\allX_{\tildeA(i,a(i))} \to \X_i), \label{eq:prf:Alg2:2} 
\eeqa %
where \eqref{eq:prf:Alg2:1} follows since $\widetilde{\calG}_K = \bigcup_{T \in \calT} \widetilde{\calG}_K^T$ and \eqref{eq:prf:Alg2:2} uses that $\PT^T$ is the best approximation in $\widetilde{\calG}_K^T$. Algorithm~2 finds the solution to \eqref{eq:prf:Alg2:2} and thus identifies the optimal approximation $\PT^* \in \widetilde{\calG}_K$. \hfill \end{proof}


\section{Proof of Theorem~\ref{thm:grd_bnd}} \label{apnd:prf:grd_bnd}

The proof is based on the proof for a related bound for submodular functions \citep{ nemhauser1978analysis}.

\begin{proof} For simplicity, we prove the case $A \bigcap B = \emptyset.$  The other case is almost identical and results in a tighter bound.   For that case the greedy algorithm selects each element of $B$ before any element of $A$. Let $l\leq |B|=L$.  Let $A_l$ be the set $A$ but ordered according to how the greedy algorithm would pick elements from $A$ after picking $\{B(1), \ldots, B(l)  \}$.

We first note two inequalities.  For all $l < L$,  
\beqa 
\I( \X_{B(l+1)} \to \Y \| \allX_{ \{B(1), \dots, B(l) \} }  )  \geq  \I( \X_{A_l(1)} \to \Y \| \allX_{ \{B(1), \dots, B(l) \} }  ),  \label{eq:proof:gensub:a1}
\eeqa which holds since the greedy algorithm selects $\X_{B(l+1)}$ after $\{\X_{B(1)}, \dots, \X_{B(l)} \}$, and 
\beqa 
\alpha^{i-1} \I( \X_{A_l(1)} \to \Y \| \allX_{ \{B(1), \dots, B(l) \} }  ) 
\geq \I( \X_{A_l(i)} \to \Y \| \allX_{ \{B(1), \dots, B(l),A_l(1), \ldots, A_l(i-1)  \} }  ),  \label{eq:proof:gensub:a2} \eeqa which follows from Assumption~\ref{assump:greedysub} for the set $A \cup \{B(1), \ldots, B(l)  \}$.

We now compare an optimal solution $A$ to the first $l$ elements in the greedy solution $B$.  
\beqa
&& \hspace{-1.65cm} \I(\allX_{A} \to \Y) - \I(\allX_{{\{B(1), \dots, B(l)\} } } \to \Y)  \nonumber \\
&&\hspace{-.2cm} \leq \I(\allX_{A \cup {\{B(1), \dots, B(l)\} } } \!\to\! \Y) - \I(\allX_{\{B(1), \dots, B(l)\} } \!\to\! \Y) \nonumber \\
&& \hspace{-.2cm} = \I(\allX_{{\{B(1), \dots, B(l)\} } } \to \Y) \nonumber \\
&& \hspace{0.1cm} + \! \sum_{i =1}^K  \I( \X_{A_l(i)} \!\to\! \Y \| \allX_{{\{B(1), \dots, B(l)\} } \cup \{A_l(1), \dots, A_l(i-1) \} } ) \nonumber \\
&& \hspace{0.1cm} - \I(\allX_{\{B(1), \dots, B(l)\} } \to \Y) \label{eq:proof:gensub:2}\\
&&\hspace{-.2cm} \leq    \sum_{i =1}^K  \alpha^{i-1} \I( \X_{A_l(1)} \to \Y \| \allX_{ \{B(1), \dots, B(l)\}}) \label{eq:proof:gensub:3} \\
&&\hspace{-.2cm} \leq  \sum_{i =1}^K  \alpha^{i-1} \I( \X_{B(l+1)} \to \Y \| \allX_{ \{B(1), \dots, B(l)\}}).   \label{eq:proof:gensub:4}
\eeqa  Equation~\eqref{eq:proof:gensub:2} follows from the chain rule applied in the order the greedy algorithm would select from $A \cup \{B(1), \dots, B(l) \}$.  Equations~\eqref{eq:proof:gensub:3} and  \eqref{eq:proof:gensub:4} follow from \eqref{eq:proof:gensub:a2} and \eqref{eq:proof:gensub:a1} respectively.  

Let $\delta_{l} :=  \I(\allX_{A} \to \Y) - \I(\allX_{{\{B(1), \dots, B(l)\} } } \to \Y).$  Then
$\delta_{l} - \delta_{l+1} = \I( \X_{B(l+1)} \to \Y \| \allX_{ \{B(1), \dots, B(l)\}})$.
  Also denote $\beta := \sum_{i =1}^K  \alpha^{i-1}$. From~\eqref{eq:proof:gensub:4} we have 
$ \delta_l \leq \beta \left( \delta_{l} - \delta_{l+1}  \right) ,$ 
which implies $\delta_{l+1} \leq \left(1 - \frac{1}{\beta}\right) \delta_l   .$  Thus 
\[\delta_{l} \leq \left(1 - \frac{1}{\beta}\right)^{l} \delta_0 \leq e^{- \frac{l}{\beta}} \delta_0.  \] The last step uses the bound $(1-p) \leq e^{-p}$, which holds for all $p$. For $0 < p < 1$, both sides are positive so the inequality is conserved if powers are taken. Since $\delta_0 = \I(\allX_{A} \to \Y) - \I(\emptyset \to \Y) = \I(\allX_{A} \to \Y),$ this gives \[  \I(\allX_{A} \to \Y) - \I(\allX_{{\{B(1), \dots, B(l)\} } } \to \Y)  \leq e^{- \frac{l}{\beta}} \I(\allX_{A} \to \Y)    ,\] which after rearranging gives the theorem. 
\end{proof}


\section{Proof of Corollary~\ref{cor:KbndL}} \label{apnd:prf:cor:KbndL}

We will prove Corollary~\ref{cor:KbndL} by first solving the following optimization problem%
\begin{eqnarray}
&\max_{\{b_1,\dots,b_K\}}& \hspace{0.3cm} \sum_{i=1}^K b_i  \label{eq:prf:KbndL:3}\\
& \hspace{0.5cm} \text{s.t.} & \hspace{0.3cm} \sum_{i=1}^L b_i \leq c \label{eq:prf:KbndL:4}\\
& & \hspace{0.0cm} 0 \leq b_i \leq \alpha b_{i-1}, \; i = 2, \ldots, m, \label{eq:prf:KbndL:5}
\end{eqnarray}
where $K$ and $L$  are integers such that $K>L$ and  $\alpha>1$ and $c>0$ are real coefficients. Let $\{b_1^*, \dots, b_K^*\}$ denote a solution to \eqref{eq:prf:KbndL:3}.  

\begin{lemma} \label{lem:prf:KbndL:1}
For any optimal solution $\{b_1^*, \dots, b_K^*\}$, \eqref{eq:prf:KbndL:4} holds with equality.
\end{lemma}
\begin{proof} The proof will follow by contradiction.  Suppose $c - \sum_{i=1}^L b_i^*>0$.  Let $\gamma:= \frac{1}{L} (c - \sum_{i=1}^L b_i^*)$. Define 
\begin{equation*}
\widetilde{b}_i := \left\{
\begin{array}{cl}
b_i^* + \gamma & \text{if } i \leq L,\\
b_i^* & \text{if } i>L.
\end{array} \right.
\end{equation*}
Note that $\sum_{i=1}^L \widetilde{b}_i = \sum_{i=1}^L b_i^* + \gamma = c$ so the first constraint is met.  Also, for $i\leq L$, $\widetilde{b}_i =b_i^* + \gamma \leq \alpha b_{i-1}^* + \alpha\gamma = \alpha\widetilde{b}_{i-1}$, so the second constraint is met.  Thus, $\{\widetilde{b}_1, \dots, \widetilde{b}_K\}$ is feasible and has a larger sum than the optimal solution, $\sum_{i=1}^K \widetilde{b}_i =L\gamma+ \sum_{i=1}^K b_i^*$, contradicting $\{b_1^*,\dots,b_K^*\}$'s optimality. \end{proof}

\begin{lemma}\label{lem:prf:KbndL:2}
For any optimal solution $\{b_1^*, \dots, b_K^*\}$, \eqref{eq:prf:KbndL:5} holds with equality.
\end{lemma}
\begin{proof}
The proof will follow by contradiction.  Suppose there is an index  $i>1$ for which $b_i^* < \alpha b_{i-1}^* $.  If $i>L$, then we can set $b_i^* \gets \alpha b_{i-1}^*$ to increase the objective function, which contradicts optimality. If $i\leq L$, replace $b^*_{i-1}$ and $b_{i}^*$ with $\widetilde{b}_{i-1}$ and $\widetilde{b}_{i}$, where $\widetilde{b}_{i-1} = \frac{b^*_{i-1} + b^*_{i}}{1+ \alpha}$ and $\widetilde{b}_{i} = \frac{\alpha(b_{i-1}^* + b_{i}^*)}{1+ \alpha}$.  Note that $b^*_{i-1} + b^*_{i} = \widetilde{b}_{i-1}+ \widetilde{b}_{i},$ and the constraints are still satisfied.  This exchange necessarily results in $b_{i+1}^* \leq \alpha b_i^*< \alpha \widetilde{b}_{i} $.  Thus, the exchange can be repeated for larger $i$ until $i=L+1$.  Then set $\widetilde{b}_{L+1} \gets \alpha \widetilde{b}_{L}$ and the objective function is necessarily increased, a contradiction. \hfill \end{proof}

We can now find the solution to the optimization problem.
\begin{lemma} \label{lem:prf:KbndL:3} The optimal solution to \eqref{eq:prf:KbndL:3} is $\sum_{i=1}^K b_i^* = c \frac{1 - \alpha^K}{1 - \alpha^L}. $ 
 
\end{lemma}

\begin{proof} By Lemmas~\ref{lem:prf:KbndL:1} and \ref{lem:prf:KbndL:2},  the constraints  \eqref{eq:prf:KbndL:4} and \eqref{eq:prf:KbndL:5} hold with equality. We can first solve for $b_1^*$, %
\beqas c &=& \sum_{i=1}^L b_i^* 
=\sum_{i=1}^L \alpha^{i-1} b_1^* \\
\Longrightarrow \; b_1^* &=& \frac{c}{\sum_{i=1}^L \alpha^{i-1}}.
\eeqas %
Solving for the value of the objective function,
\beqa
\sum_{i=1}^K b_i^* &=& \sum_{i=1}^K \alpha^{i-1} b_1^* 
=\sum_{i=1}^K \alpha^{i-1} \frac{c}{\sum_{i=1}^L \alpha^{i-1}}. \label{eq:prf:KbndL:9} 
\eeqa
Using the geometric series formula \[\sum_{i=1}^{K} \alpha^{i-1} = \sum_{i=0}^{K-1} \alpha^i = \frac{1 - \alpha^K}{1 - \alpha},\] the equation \eqref{eq:prf:KbndL:9} becomes
\beqas
\sum_{i=1}^K b_i^* &=& c \frac{\sum_{i=1}^K \alpha^{i-1}}{\sum_{i=1}^L \alpha^{i-1}}  
= c \frac{1 - \alpha^K}{1 - \alpha^L}. \nonumber
\eeqas \hfill
\end{proof}

We can now prove Corollary~\ref{cor:KbndL}.

\begin{proof} Let the elements of $A_K$ and $A_L$ be ordered according to the greedy order. Consider the worst case, with $\I(\allX_{A_K} \to \Y)$ as large as possible, given \beqa\I(\allX_{\{A_K(1), \ldots, A_K(L)\}} \to \Y) \leq \I(\allX_{A_L} \to \Y). \label{eq:prf:KbndL:1} 
\eeqa 
The inequality \eqref{eq:prf:KbndL:1} holds by definition of $A_L$ being the optimal parent set of size $L$.  Greedy-submodularity imposes another constraint.  For any $0<i<K$, 
\[  \I(\X_{A_K(i+1)} \to \Y \| \allX_{\{A_K(1), \ldots, A_K(i)\}}  )  \leq \alpha \hspace{0.09cm} \I(\X_{A_K(i)} \to \Y \| \allX_{\{A_K(1), \ldots, A_K(i-1)\}}  ).  
\]

Corollary~\ref{cor:KbndL} follows from Lemma~\ref{lem:prf:KbndL:3}, substituting   $\I(\allX_{A_L} \to \Y)$ for $c$ and $\I(\X_{A_K(i)} \to \Y \| \allX_{\{A_K(1), \ldots, A_K(i-1)\}}  )$ for $b_i$.  \hfill \end{proof}


\section{Proof for Theorem~\ref{thm:grd_connect}} \label{apnd:prf:grd_connect}

\begin{proof}
Let $T_2$ denote the MWDST picked by Algorithm~2.  For an edge $e \in \{\X_j \to \X_i : 1 \leq j \neq i \leq m\}$ in the complete graph on $m$ nodes, let $w_2(e)$ denote the weight $\I(\allX_{\wtA(i,j)} \to \X_i)$ assigned by Algorithm~2.  Define $T_4$ and $w_4(e)$ for Algorithm~4 likewise. Also, let $c :=  ( \hspace{-0.0cm} 1 - \exp( -L/(\sum_{i = 0}^{K-1} \alpha^i) )  )$.  For each edge $e$ in the complete graph, \beqa w_4(e) \geq c w_2(e), \label{eq:prf:grd_tree1}\eeqa  which follows from Theorem~\ref{thm:grd_bnd}.  Furthermore, 
\beqa \sum_{e \in T_4} w_4(e) &\geq& \sum_{e \in T_2} w_4(e) \label{eq:prf:grd_tree2}\\
&\geq& c \sum_{e \in T_2} w_2(e). \label{eq:prf:grd_tree3}
\eeqa  Equation \eqref{eq:prf:grd_tree2} follows since in Algorithm~4, $T_4$ was selected as the MWDST, and \eqref{eq:prf:grd_tree3} follows from \eqref{eq:prf:grd_tree1}.
\end{proof}


\section{Proof for Theorem~\ref{thm:apx:TopRGeneral}} \label{apnd:prf:TopRGeneral}

To prove Theorem~\ref{thm:apx:TopRGeneral}, we first show the following lemma.

\begin{lemma} \label{lem:topR} For all $1\leq l \leq r$, the $l$th best approximation has the same parent sets except one as one of the top $l-1$ solutions. 
\end{lemma}

\begin{proof}
The proof follows by induction.  The base case, with $l=1$, holds trivially as it is the only solution.  Assume that the statement of the lemma holds for some $1\leq l < r$.  Consider the $(l+1)$th best solution, $\{B(i)\}_{i=1}^m$.  Let $\X_j$ be a process for which the parent set $B(j)$ is not the same as that of the optimal solution, $A(j)$.  

By Corollary~\ref{cor:apx:gen_K_apx}, parent sets can be identified independently.  Also, by Assumption~ \ref{assump:nonequalpar}, no two parent sets have the same influence.  Thus, the optimal parent set is $A(j)$ is better than $B(j)$.  Let $A'(j)$ be any parent set for $\X_j$ that is better than $B(j)$.  Then, by Corollary~\ref{cor:apx:gen_K_apx}, the parent sets $\{B(1), \dots, B(j-1),A'(j),B(j+1),\dots,B(m)\}$ induce a better approximation than the $(l+1)$th best approximation with $\{B(i)\}_{i=1}^m$ and therefore must be one of the top $l$ approximations.  Since this new approximation differs from the $(l+1)$th best approximation in precisely one parent set, the lemma holds. \hfill
\end{proof}

The proof for Theorem~\ref{thm:apx:TopRGeneral} follows from Lemma~\ref{lem:topR}, since every approximation selected in Algorithm~5 is used as a seed in Algorithm~6 to generate all of the best solutions that have precisely one parent set different from that of the seed.


\section{Implementation Notes for Algorithm~5} \label{apnd:disc:TopRGeneral}

In Algorithm~5, for large $m$ and $r$, a naive implementation of lines~6--8 can be computationally expensive.  Specifically, redundant solutions can appear in $\calS$, and searching to remove redundancies or entries already in $Top$ might be slow.  Instead, $\calS$ can be kept as a priority queue of value-key pairs, where the value is the sum of the directed information values \eqref{eq:apx:gen:thm} for the approximation as in line~6 and the key is the index of an approximation.  There are ${m-1 \choose K}^m$ possible bounded in-degree approximations, and a binary vector can track whether an approximation has been seen or not.

We now discuss a method to compute an index for each approximation.  First, indices for individual parent sets will be identified, then combined for an index for the whole approximation.  Let $\{j_1, j_2, \dots, j_K\}$ denote the elements of parent set $A(i)$, in ascending order.  For $k \in [K]$, set $j_k \gets j_k -1$ if $j_k>i$.  Denote the set of these (possibly) modified values by the length $K$ vector $idx$.  Then run Algorithm~7.
\newcommand{\phz}{\phantom{0}}

\begin{table}[t!] \begin{normalsize} \begin{center}
\begin{tabular*}{\linewidth}{@{}llrr@{}}
{\bfseries Algorithm 7. {\sc GetParSetIndex}}\\
\hline {\bf Input:} $m, idx$ \\
\hline

~\phz 1. $cnt \gets 0$  \\
~\phz 2. $K \gets |idx|$  \\
~\phz 3. {\bf If} $K = 0$  \\
~\phz 4. \quad\  {\bf Return} $0$ \\
~\phz 5. {\bf If} $K = 1$  \\
~\phz 6. \quad\  {\bf Return} $idx(1)-1$ \\
~\phz 7. $cnt \gets \sum_{l=2}^{idx(1)} {m-l \choose K-1}$  \\
~\phz 8. $idx' \gets \{idx(2)-idx(1), \dots, idx(K)-idx(1) \}$  \\
~\phz 9. $cnt \gets cnt + \mathrm{GetParSetIndex}(m-idx(1), idx')$ \\
~10. {\bf Return} $cnt$\\
\hline
\end{tabular*}\end{center}\label{alg:GetParSetIndex} \end{normalsize}
\end{table}

Lines~7--9 in Algorithm~7 count how many parent sets of $\X_i$ are lexicographically ordered before $A(i)$.  Line~7 counts how many sets have a first element smaller than $idx(1)$.  Lines~8--9 use recursion to count how many sets with the same first element $idx(1)$ appear before $idx$.

Once the index $a_i \gets \mathrm{GetParSetIndex}(m, idx)$ for each parent set $A(i)$ of an approximation is calculated, the index for the approximation can be computed as 
\beqa
1 + \sum_{i=1}^m a_i {m-1 \choose K}^{i-1}. \nonumber
\eeqa

\vskip 0.2in
\bibliography{refs}

\end{document}